\documentclass[twoside,leqno,twocolumn]{article}
\usepackage{ltexpprt}

\usepackage{t1enc}
\usepackage{pgfplots}
\usepackage{tikz}
\usepackage{amsfonts}

\newcommand{\ceil}[1]{\left\lceil #1\right\rceil}

\newcommand{\set}[1]{\left\{ #1\right\}}
\newcommand{\gilt}{:}
\newcommand{\sodass}{\,:\,}
\newcommand{\setGilt}[2]{\left\{ #1\sodass #2\right\}}

\newcommand{\realrange}[2]{\left[#1, #2\right]}

\newcommand{\unitrange}[2]{\realrange{0}{1}}

\newcommand{\Oh}[1]{\mathcal{O}\!\left( #1\right)}

\newcommand{\llabel}[1]{\label{\labelprefix:#1}}
\newcommand{\labelprefix}{} 

\newcommand{\discussionsize}{\small}

\newcommand{\frage}[1]{}

\newenvironment{code}{\noindent
\begin{tabbing}%
\hspace{2em}\=\hspace{2em}\=\hspace{2em}\=\hspace{2em}\=\hspace{2em}\=%
\hspace{2em}\=\hspace{2em}\=\hspace{2em}\=\hspace{2em}\=\hspace{2em}\=%
\kill}{\end{tabbing}}

\newcommand{\labelcommand}{}
\newcommand{\captiontext}{}
\newsavebox{\codeparam}
\newcounter{lineNumber}
\newenvironment{disscodepos}[3]{%
\renewcommand{\labelcommand}{#2}%
\renewcommand{\captiontext}{#3}%
\sbox{\codeparam}{\parbox{\textwidth}{#3}}%
\begin{figure}[#1]\begin{center}\begin{code}\setcounter{lineNumber}{1}}{%
\end{code}\end{center}\caption{\llabel{\labelcommand}\captiontext}\end{figure}}

{\end{disscodepos}}

\newcommand{\For}      {{\bf for\ }}

\newcommand{\If}       {{\bf if\ }}
\newcommand{\Is}       {:=}

\newdimen\endofsize\endofsize=0.5em
\def\endofbeweis{~\quad\hglue\hsize minus\hsize
                 \hbox{\vrule height \endofsize width
\endofsize}\par}

\usepackage{epsfig}
\usepackage{url}
\usepackage{amsmath}
\usepackage{amssymb}
\usepackage{pdflscape}
\usepackage{latexsym}
\usepackage{todonotes}
\usepackage{hyperref}

\usepackage[algo2e,ruled,vlined]{algorithm2e}
\SetCommentSty{textit}
\DontPrintSemicolon
\IncMargin{-\parindent}
\SetAlCapHSkip{0pt}
\SetAlgoLined
\makeatletter
\newcommand{\removelatexerror}{\let\@latex@error\@gobble}
\makeatother

\newcommand{\Inc}{{++}}

\usepackage{graphicx}
\usepackage{booktabs}
\usepackage{numprint}
\npdecimalsign{.} 
\newcommand{\ie}{i.e.\ }
\newcommand{\etal}{et~al.}

\def\MdR{\ensuremath{\mathbb{R}}}

\usepackage{color}
\usepackage{tabularx}
\usepackage{subcaption}
\usepackage{dblfloatfix}

\usepackage{doi}
\usepackage[square,numbers,sort&compress]{natbib}
\newcommand{\mytitle}{Scalable Edge Partitioning}

\begin{document}

\title{\mytitle\thanks{Partially supported by DFG grant SA 933/10-2. The research leading to these results has received funding from the European Research Council under the European Union's Seventh Framework Programme (FP/2007-2013) / ERC Grant Agreement no. 340506.}}

\author{Sebastian Schlag\thanks{Institute for Theoretical Informatics, Karlsruhe Institute of Technology, Karlsruhe, Germany.} \\
\and
Christian Schulz\thanks{Faculty of Computer Science, University of Vienna, Vienna, Austria.} \\
\and
Daniel Seemaier\thanks{Institute for Theoretical Informatics, Karlsruhe Institute of Technology, Karlsruhe, Germany.} \\
\and
Darren Strash\thanks{Department of Computer Science, Hamilton College, Clinton, NY, USA.}}
\date{}

\maketitle

\fancyfoot[C]{\thepage}

\thispagestyle{fancy}

\begin{abstract} 
Edge-centric distributed computations have appeared as a recent technique to improve the shortcomings of think-like-a-vertex algorithms on large scale-free networks. In order to increase parallelism on this model, \emph{edge partitioning}---partitioning edges into roughly equally sized blocks---has emerged as an alternative to traditional (node-based) graph partitioning.
In this work, we give a distributed memory parallel algorithm to compute high-quality edge partitions  in a scalable way.
Our algorithm scales to networks with billions of edges, and runs efficiently on thousands of PEs.
Our technique is based on a fast parallelization of split graph construction and a use of advanced 
node partitioning algorithms. 
Our extensive experiments show that our algorithm 
has high quality on large real-world networks
and large hyperbolic random graphs---which have a power law degree distribution and are therefore specifically
targeted by edge partitioning.
\end{abstract}
\clearpage

\pagestyle{fancy}
\section{Introduction}
With the recent stagnation of Moore's law, the primary method for gaining computing power is to increase the number of available cores, processors, or networked machines (all of which are generally referred to as \emph{processing elements} (PEs)) and exploit parallel computation. One increasingly useful method to take advantage of parallelism is found in \emph{graph partitioning}~\cite{schloegel2000gph,GPOverviewBook,SPPGPOverviewPaper}, which attempts to partition the vertices of a graph into roughly equal disjoint sets (called \emph{blocks}), while minimizing some objective function---for example minimizing the number of edges crossing between blocks. Graph partitioning is highly effective, for instance, for distributing data to PEs in order to minimize communication volume~\cite{hendrickson2000graph}, and to minimize the overall running time of jobs with dependencies between~computation~steps~\cite{gpApplicationPaper}.

This traditional (node-based) graph partitioning has also been essential for making efficient distributed graph algorithms in the Think Like a Vertex (TLAV) model of computation~\cite{mccune-tlav-2015}. In this model, node-centric operations are performed in parallel, by mapping nodes to PEs and executing node computations in parallel. Nearly all algorithms in this model require information to be communicated between neighbors --- which results in network communication if stored on different PEs --- and therefore high-quality graph partitioning directly translates into less communication and faster overall running time. As a result, graph partitioning techniques are included in popular TLAV platforms such as Pregel~\cite{pregel-paper} and GraphLab~\cite{graphlab-paper}. 

However, node-centric computations have serious shortcomings on power law graphs --- which have a skewed degree distribution. In such networks, the overall running time is negatively affected by very high-degree nodes, which can result in more communication steps. To combat these effects, Gonzalez~\etal~\cite{gonzalez-powergraph-2012} introduced edge-centric computations, which duplicates node-centric computations across edges to reduce communication overhead. In this model, \emph{edge partitioning}---partitioning edges into roughly equally sized blocks---must be used to reduce the overall running time. 
This variant of partitioning is also NP-hard~\cite{bourse-2014}.

Similar to (node-based) graph partitioning, the quality of the edge partitioning can have a dramatic effect on parallelization~\cite{li2017spac}.
Noting that edge partitioning can be solved directly with hypergraph partitioners, such as hMETIS~\cite{hMetisRB,hMetisKway} and PaToH~\cite{PaToH}, Li~\etal~\cite{li2017spac} showed that these techniques give the highest quality partitionings; however, they are also slow.
Therefore, a balance of solution quality and speed must be taken into consideration. This balance is struck well for
the split-and-connect (SPAC) method introduced by Li~\etal~\cite{li2017spac}. In the SPAC method, vertices are duplicated and weighted so that a (typically fast) standard node-based graph partitioner can be used to compute an edge partitioning; however, this method was only studied in the sequential setting. While this is a great initial step, the graphs that benefit the most from edge partitioning are massive---and therefore do not fit on a single machine~\cite{shp-vldb}.

However, distributed algorithms for the problem fare far worse~\cite{gonzalez-powergraph-2012,bourse-2014}. While adding much computational power with many processing elements (PEs), edge partitioners such as  PowerGraph~\cite{gonzalez-powergraph-2012} and Ja-Be-Ja-VC~\cite{fatemeh2014jabejavc}, produce partitionings of significantly worse quality than those produced with hypergraph partitioners or SPAC. Thus, there is no clear winning algorithm that gives high quality while executing quickly in a distributed setting.

\subsection{Our Results.}
In this paper, we give the first \emph{high-quality} distributed memory parallel edge partitioner. Our algorithm scales to networks with billions of edges, and runs efficiently on thousands of PEs.
Our technique is based on a fast parallelization of split graph construction and a use of advanced 
node partitioning algorithms. Our experiments show that while hypergraph partitioners outperform SPAC-based graph partitioners in the sequential setting regarding both solution quality and running time,
our new algorithms compute significantly \emph{better} solutions than the distributed memory hypergraph partitioner Zoltan~\cite{Zoltan} in \emph{shorter} time.
For large random hyperbolic graphs, which have a power law degree distribution and are therefore specifically targeted by edge partitioning,
our algorithms compute solutions that are more than a factor of two better.
Moreover, our techniques scale well to \numprint{2560} PEs, allowing for efficient partitioning of graphs with billions of edges within seconds.

\section{Preliminaries}
\subsection{Basic Concepts.}
Let  $G=(V=\{0,\ldots, n-1\},E,c,\omega)$ be an undirected graph 
with edge weights $\omega: E \to \MdR_{>0}$, node weights
$c: V \to \MdR_{\geq 0}$, $n = |V|$, and $m = |E|$.
We extend $c$ and $\omega$ to sets, i.e.,
$c(V')\Is \sum_{v\in V'}c(v)$ and $\omega(E')\Is \sum_{e\in E'}\omega(e)$.
$N(v)\Is \setGilt{u}{\set{v,u}\in E}$ denotes the neighbors of $v$ and $E(v) := \{ e : v \in e \}$ denotes the edges incident to $v$.
A node $v \in V_i$ that has a neighbor $w \in V_j, i\neq j$, is a \emph{boundary node}. 
We are looking for \emph{blocks} of nodes $V_1$,\ldots,$V_k$ 
that partition $V$, i.e., $V_1\cup\cdots\cup V_k=V$ and $V_i\cap V_j=\emptyset$
for $i\neq j$. The \emph{balance constraint} demands that 
$\forall i\in \{1..k\}\gilt c(V_i) \leq L_{\max} := (1+\varepsilon)\lceil\frac{c(V)}{k}\rceil$
for some imbalance parameter $\varepsilon$. 
The objective is to minimize the total \emph{cut} $\sum_{i<j}w(E_{ij})$ where 
$E_{ij}\Is\setGilt{\set{u,v}\in E}{u\in V_i,v\in V_j}$. 
Similar to the node partitioning problem, the \emph{edge partitioning problem} asks for blocks of edges $E_1, \dots, E_k$ that partition $E$, i.e. $E_1 \cup \dots \cup E_k = E$ and $E_i \cap E_j = \emptyset$ for $i \neq j$. 
The \emph{balance constraint} demands that $\forall i \in \{1..k\}: \omega(E_i) \le (1 + \varepsilon) \ceil{\frac{\omega(E)}{k}}$.
The objective is to minimize the \emph{vertex cut} $\sum_{v \in V} |I(v)| - 1$ where $I(v) := \{ i : E(v) \cap E_i \neq \emptyset \}$. 
Intuitively, the objective expresses the number of required \emph{replicas} of nodes: 
if a node $v$ has to be copied to each edge partition that has edges incident to $v$, the number of required replicas of that node is $|I(v)| - 1$.

A \emph{clustering} is also a partition of the nodes. However, $k$ is usually not given in advance and the balance constraint is removed. 
A size-constrained clustering constrains the size of the blocks of a clustering by a given upper bound $U$ such that $c(V_i) \leq U$. 
Note that by adjusting the upper bound one can somewhat control the number of blocks of a feasible clustering. 

\subsection{Hypergraphs.}
An \textit{undirected hypergraph} $H=(V,E,c,\omega)$ is defined as a set of $n$ vertices $V$ and a
set of $m$ hyperedges/nets $E$ with vertex weights $c:V \rightarrow \mathbb{R}_{>0}$ and net 
weights $\omega:E \rightarrow \mathbb{R}_{>0}$, where each net is a subset of the vertex set $V$ (i.e., $e \subseteq V$). The vertices of a net are called \emph{pins}.
As before, $c$ and $\omega$ are extended to work on sets.
A vertex $v$ is \textit{incident} to a net $e$ if $v \in e$.
The \textit{size} $|e|$ of a net $e$ is the number of its pins. 
A \emph{$k$-way partition} of a hypergraph $H$ is a partition of its vertex set into $k$ \emph{blocks} $\mathrm{\Pi} = \{V_1, \dots, V_k\}$ 
such that $\bigcup_{i=1}^k V_i = V$, $V_i \neq \emptyset $ for $1 \leq i \leq k$ and $V_i \cap V_j = \emptyset$ for $i \neq j$.
A $k$-way partition $\mathrm{\Pi}$ is called \emph{$\mathrm{\varepsilon}$-balanced} if each block $V_i \in \mathrm{\Pi}$ satisfies the \emph{balance constraint}:
$c(V_i) \leq L_{\max} := (1+\varepsilon)\lceil \frac{c(V)}{k} \rceil$ for some parameter $\mathrm{\varepsilon}$. 
Given a $k$-way partition $\mathrm{\Pi}$, the number of pins of a net $e$ in block $V_i$ is defined as
$\mathrm{\Phi}(e,V_i) := |\{v \in V_i~|~v \in e \}|$. 
For each net $e$, $\mathrm{\Lambda}(e) := \{V_i~|~ \mathrm{\Phi}(e, V_i) > 0\}$ denotes the \emph{connectivity set} of $e$.
The \emph{connectivity} of a net $e$ is the cardinality of its connectivity set: $\mathrm{\lambda}(e) := |\mathrm{\Lambda}(e)|$.
A net is called \emph{cut net} if $\mathrm{\lambda}(e) > 1$.
The \emph{$k$-way hypergraph partitioning problem} is to find an $\varepsilon$-balanced $k$-way partition $\mathrm{\Pi}$ of a hypergraph $H$ that
minimizes an objective function over the cut nets for some $\varepsilon$.
The most commonly used cost functions are the \emph{cut-net} metric $\text{cut}(\mathrm{\Pi}) := \sum_{e \in E'} \omega(e)$ and the
\emph{connectivity} metric $(\mathrm{\lambda} - 1)(\mathrm{\Pi}) := \sum_{e\in E'} (\mathrm{\lambda}(e) -1)~\omega(e)$, where $E'$ is the set of all cut nets~\cite{UMPa,donath1988logic}.
Optimizing both objective functions is known to be NP-hard \cite{Lengauer:1990}.

\subsection{The Split-and-Connect (SPAC) Method.}\label{subs:spac}

If the input graph fits into memory, edge partitioning can be solved by forming a hypergraph with a node for each edge in $E$, and a hyperedge for each node, consisting of the edges to which it is incident. Thus, hypergraph partitioners optimizing the connectivity metric can be applied to this problem~\cite{li2017spac}; however, they are more powerful than necessary, as this conversion has mostly small hyperedges.

\begin{figure}[t]
    \centering
    \begin{subfigure}[b]{0.24\textwidth}
        \centering 
        \includegraphics[page=1]{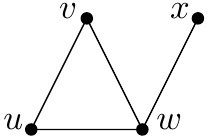}
        \caption{}
    \end{subfigure}
    \begin{subfigure}[b]{0.24\textwidth}
        \centering 
        \includegraphics[page=2]{assets/splitgraph_example}
        \caption{}
    \end{subfigure}
    \caption{(a) The input graph. (b) The resulting split graph. Each node $v$ is replaced by a set $S_v$ of split nodes that form a cycle (if $|S_v| \geq 3$). Auxiliary edges are drawn thin, dominant edges are drawn thick.}
    \label{fig:splitgraph_example}
\end{figure}

The problem can also be solved with node-based graph partitioning by creating a new graph $G'$ with the \emph{split-and-connect transformation} (SPAC) of Li~\etal~\cite{li2017spac}. 
More precisely, given an undirected, unweighted graph $G = (V, E)$, they construct the split graph $G' = (V', E', c', \omega')$ as follows: 
for each node $v \in V$, create a set of \emph{split nodes} $S_v := \{v'_1, \dots, v'_{d(v)}\}$ that are connected to a cycle by \emph{auxiliary edges} with edge-weight one, i.e. edges $\{v'_i, v'_{i + 1}\}$ for $i = 1, \dots, d(v) - 1$ and $\{v'_{d(v)}, v'_1\}$. 
In the connect phase, split nodes are connected by edges, \ie for each edge $e = \{u,v\}$ in $G$, a corresponding \emph{dominant edge} $\{u',v'\}$ in $G'$ is created.
This is done such that overall both $u' \in S_u$ and $v' \in S_v$ are connected to one and only one dominant edge.
Those dominant edges get assigned edge weight infinity. Figure~\ref{fig:splitgraph_example} gives an example.

Note that the original version of the SPAC method connects the split nodes of the same split set to an induced path rather than a cycle. Exchanging them for cycles simplifies implementation and does not change the theoretical approximation bound~\cite{li2017spac}. 

To partition the edges of $G$, a node-based partitioning algorithm is run on $G'$. 
Since dominant edges have edge weight infinity, it is infeasible to cut dominant edges for the node-based partitioning algorithm. Thus, both endpoints of a dominant edge are put into the same partition. 
To obtain an edge partition of the input graph, one transfers the block numbers of those endpoints to the edge in $G$ that~induced~the~dominant~edge.  

Each distinct node partition occurring in a split node set cuts at least one auxiliary edge, unless the set is fully contained in a single partition. 
Hence, the number of node replicas is at most the number of edge cuts.

Since the vertex cut is always smaller than or equal to the edge cut, a good node partition of the split graph intuitively leads to a good edge partition of the input graph. 
Note that the node-based partitioning algorithm is a parameter of the algorithm.
Overall, their approach is shown to be up to orders of magnitude faster than the hypergraph partitioning approaches using hMETIS~\cite{hMetisManual} and PaToH~\cite{PaToHManual} and considered competitive in terms of the number of replicas~\cite{li2017spac}.

\section{Related Work}

There has been a \emph{huge} amount of research on graph and hypergraph partitioning so that we refer the reader to existing literature \cite{schloegel2000gph,GPOverviewBook,SPPGPOverviewPaper,Papa2007} for most of the material. 
Here, we focus on issues closely related to our main contributions. Since both graph partitioning algorithms using the method by Li~\etal~\cite{li2017spac} and hypergraph partitioning algorithms are useful to solve the edge partitioning problem, we start this section with reviewing literature for those problems and then finish with edge partitioning algorithms.

\subsection{Node Partitioning.} All general-purpose methods that are able to obtain high-quality node partitions for large real-world graphs are based on the multilevel method. 
In the \emph{multilevel graph partitioning} (MGP) method, the input graph is recursively \emph{contracted} to achieve smaller graphs which should reflect the same basic structure as the input graph. After applying an \emph{initial partitioning} algorithm to the smallest graph, the contraction is undone and, at each level, a
\emph{local search} method is used to improve the partitioning induced by the coarser level.
Well-known software packages based on this approach include KaHIP~\cite{kabapeE}, Jostle~\cite{Walshaw07}, METIS~\cite{karypis1998fast} and Scotch \cite{ptscotch}.

Most probably the fastest available parallel code is the parallel version of METIS, ParMETIS~\cite{karypis1996parallel}.
This parallelization has difficulty maintaining the balance of the partitions since at
any particular time, it is difficult to say how many nodes are assigned to a
particular block. 
PT-Scotch \cite{ptscotch}, the parallel version of Scotch, is based on
recursive bipartitioning.

Within this work, we use sequential and distributed memory parallel algorithms of the open source multilevel graph partitioning framework KaHIP~\cite{kabapeE} (Karlsruhe High Quality Partitioning).
This framework tackles the node partitioning problem using the edge cut as objective.
ParHIP~\cite{parhip} is a distributed memory parallel node partitioning algorithm. The algorithm is based on parallelizing and adapting the \emph{label propagation}   technique originally developed for graph clustering~\cite{labelpropagationclustering}.  By introducing size   constraints, label propagation becomes applicable for both the coarsening and  the refinement phase of multilevel graph partitioning. 
   The resulting system is more scalable and achieves higher quality  than state-of-the-art systems like ParMETIS and PT-Scotch.

\subsection{Hypergraph Partitioning.}
HGP has evolved into a broad research area since the 1990s.
Well-known multilevel HGP software packages with certain distinguishing characteristics include PaToH~\cite{PaToH} (originating from scientific computing),
hMETIS~\cite{hMetisRB,hMetisKway} (originating from VLSI design), KaHyPar~\cite{ahss2017alenex,hs2017sea,KaHyPar-R} (general purpose, $n$-level), Mondriaan~\cite{Mondriaan} (sparse matrix
partitioning), MLPart~\cite{MLPart} (circuit partitioning), Zoltan~\cite{Zoltan}, and SHP~\cite{shp-vldb}  (distributed),
UMPa~\cite{DBLP:conf/dimacs/CatalyurekDKU12} (directed hypergraph model, multi-objective), and kPaToH (multiple constraints, fixed vertices)~\cite{Aykanat:2008}.

\subsection{Edge Partitioning.}
While hypergraph partitioning and the SPAC method are effective for computing an edge partitioning of small graphs, different techniques are used for graphs that do not fit in the memory of a single computer.
Gonzalez \etal~\cite{gonzalez-powergraph-2012} study the streaming case, where edges are assigned to partitions in a single pass over the graph. They investigate randomly assigning edges to partitions, as well as a greedy strategy. 
Using their greedy method, the average number of replicas is around $5$ for power-law graphs.
Bourse \etal~\cite{bourse-2014} later improved the replication factor by performing a similar process, but weighting each vertex by its degree. 
With Ja-Be-Ja-VC, Rahimian \etal~\cite{fatemeh2014jabejavc} present a distributed parallel algorithm for edge-based partitioning of large graphs which is essentially a local search algorithm
that iteratively improves upon an initial random assignment of edges to
partitions. 

\section{Engineering a Parallel Edge Partitioner}\label{sec:edge_partitioner}
We now present our algorithm to quickly compute high quality edge partitions in the distributed memory setting.
Roughly speaking, we engineer a distributed version of the SPAC algorithm (dSPAC) and then use a distributed memory parallel node-based graph partitioning to partition the model.
We start this section by giving a description of the data structures that we use and then explain the parallelization of the SPAC algorithm.

\subsection{Graph Data Structure.} We start with details of the parallel graph data structure and the implementation of the methods that handle communication. 
First of all, each PE  gets a subgraph, \ie a contiguous range of nodes $a..b$, of the whole graph as its input, such that the subgraphs combined correspond to the input graph.
Each subgraph consists of the nodes with IDs from the interval $I:=a..b$ and the edges incident to at least one node in this interval, and includes vertices not in $I$ that are adjacent to vertices in $I$. These vertices are referred to as \emph{ghost} nodes (also referred to elsewhere in the literature as \emph{halo} nodes). Note that each PE may have edges it shares with another PE and the number of edges assigned to the PEs may therefore vary significantly.
The subgraphs are stored using a standard adjacency array representation---we have one array to store edges and one array for nodes, which stores head pointers into the edge array. 
However, for the parallel setting, the node array is divided into two parts:
the first part stores local nodes and the second part stores ghost nodes. 
The method used to keep local node IDs and ghost node IDs consistent is explained next.

Instead of using the node IDs provided by the input graph (i.e., the \emph{global} IDs), each PE $p$ maps those IDs to the range $0\, .. \, n_p-1$, where $n_p$ is the number of distinct nodes of the subgraph on that PE. Note that this number includes the number of ghost nodes stored on PE $p$. The number of local nodes is denoted by $\ell_p$. 
Each global ID $i \in a \, .. \, b$ is mapped to a local node ID $i-a$. The IDs of the ghost nodes are mapped to the remaining $n_p - (b-a + 1)$ local IDs in the order in which they appeared during the construction of the graph data structure. 
Transforming a local node ID to a global ID or vice versa, can be done by adding or subtracting $a$. 
We store the global ID of the ghost nodes in an extra array and use a (local) hash table to transform global IDs of ghost nodes to their corresponding local IDs.
Additionally, we store for each ghost node the ID of its corresponding PE, using an array $\FuncSty{PE}$ for $\mathcal{O}(1)$ lookups. 
We call a node an \emph{interface node} if it is adjacent to at least one ghost node. The PE associated with the ghost node is called an \emph{adjacent PE}.

Similar to node IDs, directed edges are mapped to the range $0\, .. \, m_p-1$ of local edge IDs, where $m_p$ denotes the number of local directed edges of the subgraph on PE $p$. 
Edges that start from node $v \in V(G)$ have consecutive edge IDs $e_v \, .. \, e_{v + d(v) - 1}$.
A roughly equal distribution of nodes to PEs is suboptimal for the SPAC algorithm since the size of the split graph mostly depends on the number of edges. 
Hence, we improve load imbalance by distributing the graph such that the \emph{number of edges} is roughly the same~on~each~PE. 

\subsection{Distributed Split Graph Construction.}
We use $G$ to denote the input graph with $n_p$ nodes and $m_p$ edges on PE $p$ and construct the split graph $G'$ with $n'_p$ nodes and $m'_p$ edges on the same PE.
A pseudocode description of our distributed SPAC algorithm (dSPAC) is given in Algorithm~\ref{alg:dspac}.

First, recall that the split graph contains a set~$S_v$ of~$d(v)$ split nodes for each node $v \in V(G)$. 
For a node~$v$, we create the nodes of the split graph on the PE that owns $v$.
Thus, the number of local split graph nodes on PE $p$ is equal to the number of local edges of the input graph on PE $p$, \ie $n'_p = m_p$.  
Since the edges incident to the same node have consecutive IDs, we can use edge IDs from $G$ as local node IDs in the split graph $G'$.
To transform those local node IDs to global ones, we need the overall number of edges on PEs that have a smaller PE ID than $p$.
This can be easily computed in parallel by computing a prefix sum over the number of local edges $m_p$ of each PE $p$.
This takes $\Oh{\log p}$ time and linear work~\cite{kumar1994introduction}.
\begin{algorithm2e}[t!]
\caption{Our dSPAC algorithm on PE $p$}\label{alg:dspac}\normalsize
\KwIn{Graph $G$ with $n_p$ nodes, $m_p$ edges on PE $p$}
\KwIn{Empty split graph $G'$}
 $n'_p := m_p$ \tcp*[f]{one split node per local edge} \\
 $ m^\text{global}_p := \FuncSty{prefixSum}(m_p)$ \tcp*[f]{i.e., $\sum_{p'= 0}^{p-1} m_p$}\\
 \For(\tcp*[f]{create split nodes ... }) {$i:=0;~i < n'_p;~\Inc i$} {
   $G'.\FuncSty{insertNode}( m^\text{global}_p  + i)$\tcp*[f]{... with global IDs}
   }

   \textit{// compute $\mathcal{E}_p'$ values for all interface nodes}\\
   \ForEach(\tcp*[f]{i.e., $u \in 0\,..\, l_p - 1$}) {$u \in V(G)$} {
   $p' := -1$ \\
   \ForEach(\tcp*[f]{i.e., $ e \in e_u \, .. \, e_{u + d(u) - 1}$}) {$e \in E(u)$} {
     $v := \FuncSty{edgeTarget}(e)$ \\
     \If() {$\FuncSty{PE}[v] \neq p'$} {
        $p' := \FuncSty{PE}[v]$ \\
        $\mathcal{E}_{p'}(a + u) := m^\text{global}_p  + e$ \\
     }
   }
 }
 \textit{// exchange $\mathcal{E}_p'$ values with adjacent PEs}\\
   \ForEach() {PE $ p' \neq p$} {
     $\FuncSty{send}~\mathcal{E}_{p'}$ values to PE $p'$\\
     $\FuncSty{receive}~\mathcal{E}_p$ values from PE $p'$ \\
   }
   \textit{// insert auxiliary and dominant edges}\\
   \ForEach(\tcp*[f]{i.e., $u \in 0\,..\, l_p - 1$}) {$u \in V(G)$} {
     \ForEach(\tcp*[f]{i.e., $ e \in e_u \, .. \, e_{u + d(u) - 1}$}) {$e \in E(u)$} {
       $v := \FuncSty{edgeTarget}(e)$ \\
       $v' := \FuncSty{globalID}(v)$ \\
       $u' := e + m^\text{global}_p$ \tcp*[f]{global ID of split node $u'$} \\
       \textit{// insert dominant edge with edge-weight $\infty$}\\
       $G'.\FuncSty{insertEdge}(u', \mathcal{E}_p(v'), \infty)$\\
       $\mathcal{E}_p(v') := \mathcal{E}_p(v') + 1$ \\
       \If(\tcp*[f]{insert auxiliary edges}){$d(u) > 1$} {
         \textit{// compute target nodes for aux. edge}\\
       	$u_\text{prev} := e - e_u - 1 \mod d(u)$ \\
       	$u_\text{next} := e - e_u + 1 \mod d(u)$ \\
        \textit{// local$~\rightarrow~$global ID for target node}\\
        $v' := e_u + m^\text{global}_p + u_\text{next}$ \\
       	$G'.\FuncSty{insertEdge}(u',v',  1)$ \\
       	\If{$u_\text{prev} \neq u_\text{next}$}{
          \textit{// local$~\rightarrow~$global ID for target node}\\
          $v' := e_u + m^\text{global}_p + u_\text{prev}$\\
          $G'.\FuncSty{insertEdge}(u',v' , 1)$} 
	   }
     }
   }
   \KwOut{distributed split graph $G'$}
\end{algorithm2e}

Afterwards, the SPAC transformation requires us to connect split nodes by auxiliary edges and dominant edges.
Auxiliary edges with edge-weight one are inserted between split nodes $v \in S_v$ to connect them to an induced cycle.
Since $S_v$ is fully contained on a single PE, we can create the auxiliary edges independently~on~each~PE. 

Recall that dominant edges are induced by undirected edges $\{u, v\}$ in $G$. These edges connect split nodes, \ie nodes from the sets $S_u$ and $S_v$, such that each split node is incident to precisely one (undirected) dominant edge. To construct dominant edges, global coordination is required: say that $u$ has neighbors on two different PEs. 
Both neighbors do not know about each other, since $u$ is only available as ghost node on those PEs, \ie without information about its adjacency. Yet, both neighbors must choose a unique split node~from~the~set~$S_u$. 

\begin{figure}
    \centering
    \includegraphics{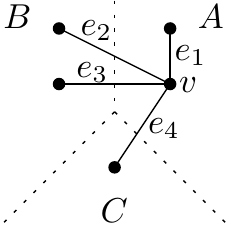}
    \caption{A node $v \in V(G)$ on PE $A$ with neighbors on PEs $A$, $B$ and $C$. Neighbors are traversed in the order of the edges. PEs are ordered $A < B < C$.}
    \label{fig:neighbors_on_pes}
\end{figure}
\begin{figure*}
    \centering
    \begin{subfigure}[c]{0.3\textwidth}
        \centering 
        \includegraphics{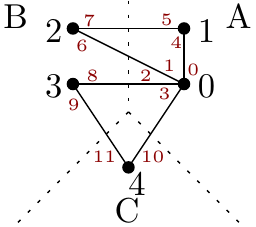}
        \caption{}
    \end{subfigure}
    \begin{subfigure}[c]{0.33\textwidth}
        \scriptsize
        \centering
        \begin{tabularx}{\columnwidth}{cXXXX}
            \toprule
            $\Rsh$ & A & B & C \\
            \midrule 
            A & & $\mathcal{E}_B(0)=1$ \newline $\mathcal{E}_B(1)=5$ & $\mathcal{E}_C(0)=3$ \\
            \midrule 
            B & $\mathcal{E}_A(2)=5$ \newline $\mathcal{E}_A(3)=7$ & & $\mathcal{E}_C(3)=8$ \\
            \midrule 
            C & $\mathcal{E}_A(4)=9$ & $\mathcal{E}_B(4)=10$ & \\
            \bottomrule
        \end{tabularx}
        \caption{}
    \end{subfigure}    
    \begin{subfigure}[c]{0.3\textwidth}
        \centering
        \includegraphics{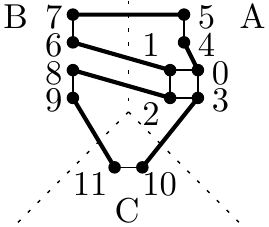}
        \caption{}
    \end{subfigure}

    \caption{(a) The input graph, distributed across three PEs (indicated by the dashed delimiters). IDs of directed edges are drawn next to their sources. (b) The $\mathcal{E}_p(\cdot)$ messages each PE sends to other PEs: the messages in the $i$-th row and $j$-th column are sent from PE $i$ to PE $j$. (c) The constructed split graph. Auxiliary edges are drawn thin, dominant edges are drawn thick.}
    \label{fig:splitgraph_construction_distributed}
\end{figure*}

Our algorithm solves this problem as follows. 
First, the adjacency lists of each node are ordered by their global node ID. 
We assume that this is already the case for the input network, otherwise one can simply run a sorting algorithm on the neighborhood of each vertex.
Since nodes are assigned consecutively among the processors, this implies that the adjacency list of each interface node in the input graph is ordered by the target processor,
 \ie the processor that owns the target of the edge.
Figure~\ref{fig:neighbors_on_pes} gives an~example.

For an interface node $u \in V(G)$ on PE $p$, let~$\mathcal{E}_{p'}(u)$  be the global ID of the first edge $\{u, v\}$ with~$v \in V(G)$ on PE $p'$.
In the split graph $G'$ this ID corresponds to the global ID of the first split node in $S_u$ that will be adjacent to a split node of $S_v$ on PE $p'$.
Due to the order of the vertices, these values can be easily computed for all adjacent PEs by scanning the neighborhood of that vertex.
We send the corresponding value to PE~$p'$. Using this information, we will be able to construct dominant edges in the desired way. 
To avoid startup overheads, we first compute all $\mathcal{E}_{p'}(v_I)$ values for all adjacent interface nodes $v_I$ and then send a \emph{single} message
from $p$ to $p'$ that contains all values.
Hence, the total message size for PE~$p$ is $\mathcal{O}(n^p_I P^p_A)$ where $n^p_I$ is the number of interface nodes on PE $p$ and~$P^p_A$ is the number~of~PEs adjacent to PE $p$.

We now create dominant edges by using the just computed values as follows. First of all, each processor traverses its nodes in the order of their IDs. 
Let $u \in V(G)$ be a node on PE $p$ with its ordered neighbors being $v_1, \dots, v_{d(v)}$. For each edge $\{u, v_i\}$, we create a dominant edge from $u$'s $i$-th split node to $\mathcal{E}_p(v_i)$, \ie we create the dominant edge $\{u'_i, \mathcal{E}_p(v_i)\}$.
Note that the value $\mathcal{E}$ was initially sent from the PE that contains~$v_i$. Afterwards, we increment $\mathcal{E}_p(v_i)$ by one, so that the next neighbor of $v_i$ that is on PE $p$ connects a dominant edge to $v_i$'s next split node. 

\begin{lemma} Our parallel SPAC algorithm creates a valid split graph.
\end{lemma}
\begin{proof}
First note that, by the process above, we create precisely one dominant edge for each split node. 
Hence, it is sufficient to show that the resulting split graph is undirected. Consider a pair of adjacent nodes $u, v \in V(G)$ where~$u$ is owned by PE $p$ and $v$ is owned by PE $p'$. 
We show that both vertices pick~the~correct~split~node --- hence, forming an undirected (dominant) edge.

Roughly speaking, by the order in which the nodes and their incident edges are traversed it is ensured that the values $\mathcal{E}$ used for the creation of the edges point to the correct split node.
More precisely, let $u \in V(G)$ be a node of the input graph and its ordered neighbors be $v_1, \dots, v_{d(u)}$. We consider~$\{u, v_i\}$ and argue that its induced dominant edge is indeed undirected. 
For this purpose, let $u$ be owned by PE $p$ and $v_i$ be owned by PE $p'$.
On PE $p$, we create a directed edge $(u'_i, v')$, where $u'_i$ is $u$'s $i$-th split node and $v'$ is \emph{some} split node of $v$ defined by the process above.
We argue that PE $p'$, creates an edge $(v', u'_i)$ which makes the graph undirected. 
It is sufficient to argue that both edges include $u'_i$, since we can use the same argument with $u$ and $v_i$ reversed to imply that the other endpoint is correct too.
Node $u$ chooses $u'_i$ for the dominant edge from $S_u$ to $S_{v_i}$, because it traverses its neighbors in order and $v_i$ is its $i$-th neighbor.
In the other direction, $v_i$ chooses the split node of $u$ based on $\mathcal{E}_{p'}(u)$.
We claim that it chooses $u'_i$. 
Let $v_j$ be the first neighbor of $u$ on PE $p'$. 
Since the neighborhood is ordered as described above, $v_j, \dots, v_i$ are all on PE $p'$ and moreover, they are traversed in the same order on PE $p'$ (thus construct their dominant edges in the same order).
Thus, $v_i$ connects the dominant edge to $\mathcal{E}_{p'}(u) + (i - j)$, since that is the total increment of $\mathcal{E}_{p'}(u)$ at the time when $v_i$ constructs its dominant edges. 
But by the definition of $\mathcal{E}_{p'}(u)$, we have that $\mathcal{E}_{p'}(u)$ is the global split node ID of $u'_j$. Thus, $v_i$ connects to $u$'s $j + (i - j) = i$-th split~node. 

\end{proof}

Assuming that the adjacency list of the nodes are already sorted by global ID, our algorithm performs a linear amount of work. Thus split graph construction takes $\mathcal{O}(m/p + \log p)$ time, if edges are~distributed~evenly. 

After computing the split graph, we use the distributed parallel node-partitioning algorithms ParHIP~\cite{parhip} and ParMETIS~\cite{karypis1996parallel} to partition it.
To obtain an edge partition of the input graph, we transfer the block numbers of those endpoints to the edge in $G$ that~induced~the~dominant~edge.  

\section{Experimental Evaluation}
\label{sec:experiments}

In this section we evaluate the performance of the proposed algorithm. We start by presenting our methodology and setup, the system used for the evaluation and the benchmark set that we used.
We then look at solution quality, running time, and scalability of (d)SPAC-based GP as well as HGP, and compare our algorithm to those systems.

\subsection{Methodology and Setup.}

We implemented the distributed split graph construction algorithm described in Section~\ref{sec:edge_partitioner} in the ParHIP graph partitioning framework~\cite{parhip}.
In the following, we use SPAC when referring to sequential split graph construction and use dSPAC to denote our algorithm in the distributed setting.
The code is written in C++, compiled with g++ 7.3.0, and uses OpenMPI 1.10 as well as KaHIP v2.0. 

In order to establish the state-of-the-art regarding edge partitioning, we perform a large number of experiments using several partitioning tools including sequential and distributed graph
and hypergraph partitioners.
More precisely, our experimental comparisons use the KaHIP~\cite{kaffpa} and METIS~\cite{karypis1998fast} sequential graph partitioners
as well as their respective distributed versions ParHIP~\cite{parhip} and ParMETIS~\cite{karypis1996parallel}. Furthermore we
use the $k$-way (hMETIS-K) and the recursive bisection variant (hMETIS-R) of hMETIS 2.0 (p1)~\cite{hMetisRB,hMetisKway}, PaToH~\cite{PaToH}, and KaHyPar-MF~\cite{hss2018sea}. These hypergraph partitioners were chosen because they provide the best solution quality for sequential hypergraph partitioning~\cite{hss2018sea}. To evaluate distributed hypergraph partitioning approaches, we include Zoltan~\cite{Zoltan}. We also tried to use Par$k$way~\cite{Parkway2.0Impl}, but were not able work with the current version provided online\footnote{\url{https://github.com/parkway-partitioner/parkway}}, because the code has deadlocks and hangs on many instances.
Since there is no implementation of the Ja-Be-Ja-VC algorithm~\cite{fatemeh2014jabejavc} publicly available, we include our own implementation.
Judging from the results presented in \cite{fatemeh2014jabejavc}, both implementations provide comparable solution quality. However, 
since Ja-Be-Ja-VC performed significantly worse than all other partitioning approaches in our experiments, we only consider it in a sequential setting.
Furthermore we do not report running times, because all other systems are highly engineered, while our  Ja-Be-Ja-VC implementation is a prototype.
For partitioning, we use $\varepsilon = 0.03$ as imbalance factor for all tools except h-METIS-R, which treats the imbalance parameter differently.
We therefore use an adjusted imbalance value as described in \cite{shhmss2015kwayvianlevel}.

For each algorithm, we perform five repetitions with different
seeds and use the arithmetic mean to average solution quality and running time of the different runs. 
 When averaging over different instances, we use the  \emph{geometric mean} in order to give every instance a comparable influence on the final result.

We furthermore use \emph{performance plots}~\cite{KaHyPar-R} to compare the best solutions of competing algorithms on a per-instance basis.
For each algorithm, these plots relate the smallest vertex cut of all algorithms to the corresponding vertex cut produced by the algorithm on a per-instance basis.
A point close to one indicates that the partition produced by the corresponding algorithm was considerably worse than the
partition produced by the best algorithm. A value of zero therefore indicates that the corresponding algorithm produced the best solution.
Thus an algorithm is considered to outperform another algorithm if its corresponding ratio values are below those of the other algorithm.

\subsection{System and Instances.} 
We use the ForHLR II cluster (Forschungshochleistungsrechner) for our experimental evaluation.
The cluster has 1152 compute nodes, each of which is equipped with 64 GB main memory and two Intel Xeon E5-2660 Deca-Core v3 processors (Haswell) clocked at 2.6 GHz. 
A single Deca-Core processor has 25 MB L3-Cache, and every core has 256 KB L2-Cache and 64 KB L1-Cache.
All cluster nodes are connected by an InfiniBand 4X EDR interconnect.

We evaluate the algorithms on the graphs listed in Table~\ref{tab:instances}.
Random geometric rggX graphs have $2^X$ nodes and were generated using code from \cite{kappa}. Random hyperbolic rhgX graphs are generated using \cite{kagen} with power law exponent $2.2$ and average degree $8$. 
SPMV graphs are bipartite locality graphs for sparse matrix vector multiplication (SPMV), which were also used to evaluate the sequential SPAC algorithm in \cite{li2017spac}. Given a  $n \times n$ matrix $M$ (in our case the adjacency matrix of the corresponding graph), an SPMV graph corresponding to an SPMV computation $Mx = y$ consists of $2 n$ vertices representing the $x_i$ and $y_i$ vector entries and contains an edge $(x_i, y_j)$ if $x_i$ contributes to the computation of $y_j$, i.e. if $M_{ij} \neq 0$.

To evaluate the hypergraph approaches, we transform the graphs into hypergraphs. As described in Section~\ref{subs:spac}, a hypergraph instance contains one hypernode for each undirected edge in the graph and a hyperedge for each graph node that contains
the hypernodes corresponding to its incident edges.

\begin{table}
    \footnotesize
	\centering
	\begin{tabular}[4]{|l|r|r|c|c|}
		\hline
		Graph & $n$ & $m$ & Type & Ref. \\
		\hline
		\hline
		\multicolumn{5}{|c|}{Walshaw Graph Archive} \\
		\hline 
		\texttt{add20} & $\approx 2.3$K & $\approx 7.4$K & M & \cite{walshaw2000mpm} \\
		\texttt{data} & $\approx 2.8$K & $\approx 15$K & M & \cite{walshaw2000mpm} \\
		\texttt{3elt} & $\approx 4.7$K & $\approx 13.7$K & M & \cite{walshaw2000mpm} \\
		\texttt{uk} & $\approx 4.8$K & $\approx 6.8$K & M & \cite{walshaw2000mpm} \\
		\texttt{add32} & $\approx 4.9$K & $\approx 9.4$K & M & \cite{walshaw2000mpm} \\
		\texttt{bcsstk33} & $\approx 8.7$K & $\approx 291$K & M & \cite{walshaw2000mpm} \\
		\texttt{whitaker3} & $\approx 9.8$K & $\approx 289$K & M & \cite{walshaw2000mpm} \\
		\texttt{crack} & $\approx 10$K & $\approx 30$K & M & \cite{walshaw2000mpm} \\
		\texttt{wing\_nodal} & $\approx 10$K & $\approx 75$K & M & \cite{walshaw2000mpm} \\
		\texttt{fe\_4elt2} & $\approx 11$K & $\approx 32$K & M & \cite{walshaw2000mpm} \\
		\texttt{vibrobox} & $\approx 12$K & $\approx 165$K & M & \cite{walshaw2000mpm} \\
		\texttt{bcsstk29} & $\approx 13$K & $\approx 302$K & M & \cite{walshaw2000mpm} \\
		\texttt{4elt} & $\approx 15$K & $\approx 45$K & M & \cite{walshaw2000mpm} \\
		\texttt{fe\_sphere} & $\approx 16$K & $\approx 49$K & M & \cite{walshaw2000mpm} \\
		\texttt{cti} & $\approx 16$K & $\approx 48$K & M & \cite{walshaw2000mpm} \\
		\texttt{memplus} & $\approx 17$K & $\approx 54$K & M & \cite{walshaw2000mpm} \\
		\texttt{cs4} & $\approx 22$K & $\approx 43$K & M & \cite{walshaw2000mpm} \\
		\texttt{bcsstk30} & $\approx 28$K & $\approx 1$M & M & \cite{walshaw2000mpm} \\
		\texttt{bcsstk31} & $\approx 35$K & $\approx 572$K & M & \cite{walshaw2000mpm} \\
		\texttt{fe\_pwt} & $\approx 36$K & $\approx 144$K & M & \cite{walshaw2000mpm} \\
		\texttt{bcsstk32} & $\approx 44$K & $\approx 985$K & M & \cite{walshaw2000mpm} \\
		\texttt{fe\_body} & $\approx 45$K & $\approx 163$K & M & \cite{walshaw2000mpm} \\
		\texttt{t60k} & $\approx 60$K & $\approx 89$K & M & \cite{walshaw2000mpm} \\
		\texttt{wing} & $\approx 62$K & $\approx 121$K & M & \cite{walshaw2000mpm} \\
		\texttt{brack2} & $\approx 62$K & $\approx 366$K & M & \cite{walshaw2000mpm} \\
		\texttt{finan512} & $\approx 74$K & $\approx 261$K & M & \cite{walshaw2000mpm} \\
		\texttt{fe\_tooth} & $\approx 78$K & $\approx 452$K & M & \cite{walshaw2000mpm} \\
		\texttt{fe\_rotor} & $\approx 99$K & $\approx 662$K & M & \cite{walshaw2000mpm} \\
		\texttt{598a} & $\approx 110$K & $\approx 741$K & M & \cite{walshaw2000mpm} \\
		\texttt{fe\_ocean} & $\approx 143$K & $\approx 409$K & M & \cite{walshaw2000mpm} \\
		\texttt{144} & $\approx 144$K & $\approx 1$M & M & \cite{walshaw2000mpm} \\
		\texttt{wave} & $\approx 156$K & $\approx 1$M & M & \cite{walshaw2000mpm} \\
		\texttt{m14b} & $\approx 214$K & $\approx 1.6$M & M & \cite{walshaw2000mpm} \\
		\texttt{auto} & $\approx 448$K & $\approx 3.3$M & M & \cite{walshaw2000mpm} \\
		\hline
		rhgX & $2^{10}$ -- $2^{18}$ & $\approx 3.6K$ -- $976$K & S & \cite{kagen} \\
		\hline 
		\hline 
		\multicolumn{5}{|c|}{SPMV Graphs} \\
		\hline
		\texttt{cant\_spmv} & $\approx 125$K & $\approx 2$M & SP & \cite{williams2009178} \\
		\texttt{scircuit\_spmv} & $\approx 350$K & $\approx 100$K & SP & \cite{UFsparsematrixcollection} \\
        \texttt{mc2depi\_spmv} & $\approx 1$M & $\approx 2.1$M & SP & \cite{williams2009178} \\		
		\texttt{in-2004\_spmv} & $\approx 2.5$M & $\approx 17$M & SP & \cite{webgraphWS} \\
		\texttt{circuit5M\_spmv} & $\approx 11$M & $\approx 60$M & SP & \cite{UFsparsematrixcollection} \\ 
		\hline 
		\hline 
		\multicolumn{5}{|c|}{Large Graphs} \\
		\hline
		\texttt{amazon} & $\approx 407$K & $\approx 2.3$M & S & \cite{snap} \\
		\texttt{eu-2005} & $\approx 862$K & $\approx 16.1$M & S & \cite{benchmarksfornetworksanalysis} \\
		\texttt{youtube} & $\approx 1.1$M & $\approx 2.9$M & S & \cite{snap} \\
		\texttt{in-2004} & $\approx 1.4$M & $\approx 27$M & S & \cite{benchmarksfornetworksanalysis} \\
		\texttt{packing} & $\approx 2.1$M & $\approx 17.4$M & M & \cite{benchmarksfornetworksanalysis} \\
		\texttt{channel} & $\approx 4.8$M & $\approx 42.6$M & M & \cite{benchmarksfornetworksanalysis} \\
		\texttt{road\_central} & $\approx 14$M & $\approx 34$M & R & \cite{benchmarksfornetworksanalysis} \\
		\texttt{hugebubble-10} & $\approx 18.3$M & $\approx 27.5$M & M & \cite{benchmarksfornetworksanalysis} \\
		\texttt{uk-2002} & $\approx 18.5$M & $\approx 262$M & S & \cite{webgraphWS} \\
		\texttt{nlpkkt240} & $\approx 27.9$M & $\approx 373$M & M & \cite{UFsparsematrixcollection} \\
		\texttt{europe\_osm} & $\approx 51$M & $\approx 108$M & R & \cite{benchmarksfornetworksanalysis} \\
          \hline
          rhgX & $2^{20}$ -- $2^{26}$ & $\approx 4$M -- $280$M & S & \cite{kagen} \\
          \hline
                            \hline
		\multicolumn{5}{|c|}{Huge Graphs} \\
		\hline
		rggX & $2^{25}$ -- $2^{28}$ & $\approx 550$M -- $5$G & M & \cite{kappa} \\
		\hline
	\end{tabular}
	\caption{Our benchmark set. Type 'S' stands for social or web graphs, 'M' is used for mesh type networks, 'R' is used for road networks, SP(MV) is used for graphs for sparse matrix-vector multiplication.}
	\label{tab:instances}
\end{table}

\subsection{Solution Quality of SPAC+X and HGP.}
\label{sec:comparison}
We start by exploring the solution quality provided by the different \emph{sequential} algorithmic approaches to the edge partitioning problem, i.e., we consider the vertex cut that
is obtained by applying the partition of the SPAC or hypergraph model to the input graph.

We restrict the benchmark set to the Walshaw graphs, SPMV graphs with up to $1$M nodes\footnote{\texttt{scircuit\_spmv}, \texttt{cant\_spmv} and \texttt{mc2depi\_spmv}} and \texttt{rhg10} -- \texttt{rhg18},
since the running times for hypergraph partitioners were too high for larger instances. We run all partitioners on \emph{one} PE, i.e., one core of a single node.
Each instance is partitioned into $k$ blocks for $k \in \{2,4,8,16,32,64,128\}$.

In the experiments of Li \etal~\cite{li2017spac}, the SPAC approach combined with METIS as graph partitioner was significantly faster than
the hypergraph partitioners hMetis and PaToH, while achieving comparable solution quality.
Since this comparison was restricted to five graphs, we first compare a larger number of high quality graph and hypergraph partitioners on a larger benchmark set.

\begin{figure}[h!]
    \centering
    \includegraphics[width=\columnwidth]{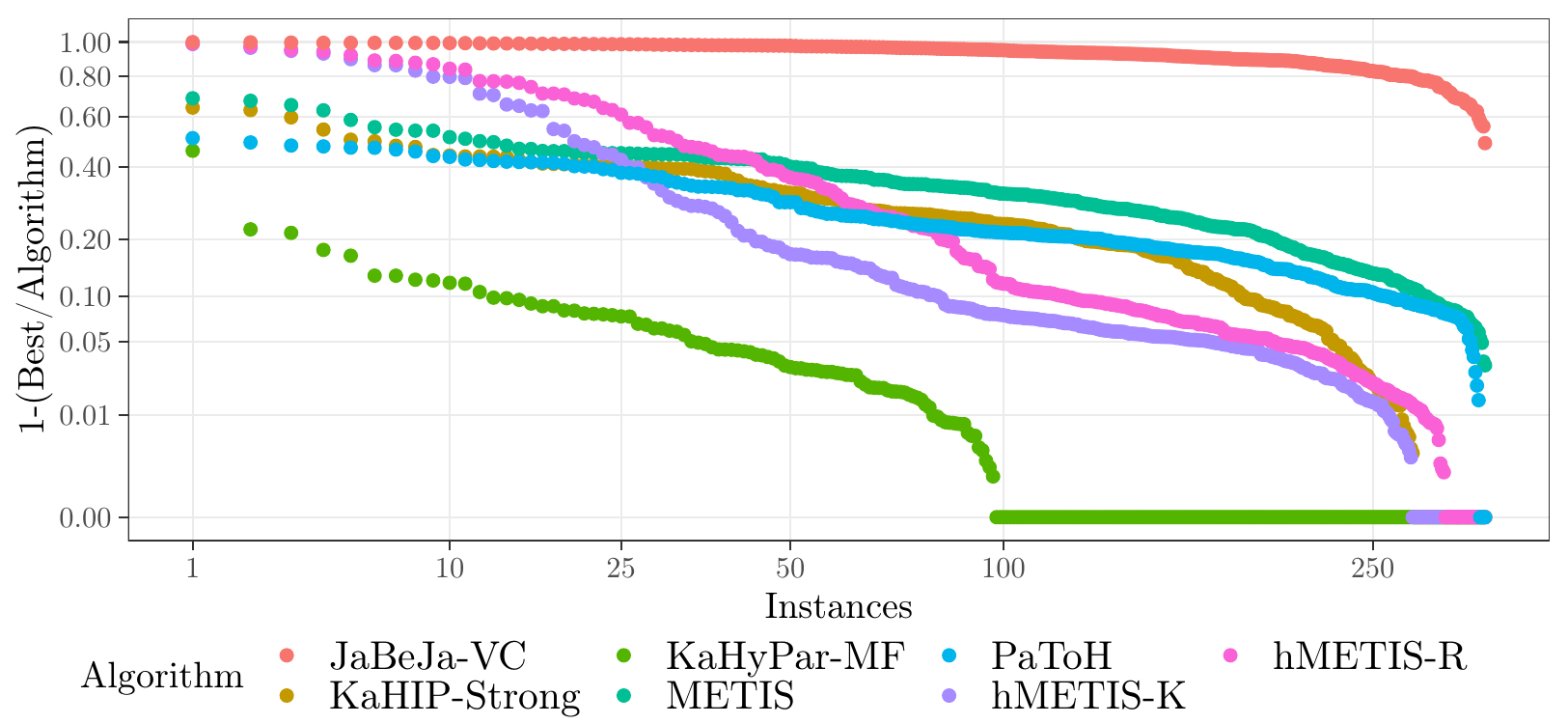}
    \caption{Performance plot comparing SPAC+X, several HGP systems, and Ja-Be-Ja-VC on small graphs.}
    \label{fig:plot_li}
\end{figure}
\begin{table}[h!]
    \begin{tabularx}{\columnwidth}{lXrrXrr}
        \toprule
         Type && \multicolumn{2}{c}{\emph{S} (rhgX)} && \multicolumn{2}{c}{\emph{M} \& \emph{SP}} \\
        Algorithm&& \multicolumn{1}{c}{VC} & \multicolumn{1}{c}{Time} && \multicolumn{1}{c}{VC} & \multicolumn{1}{c}{Time} \\
        \midrule
        KaHyPar-MF && \numprint{433} & \numprint{510.61} s && \numprint{1350} & \numprint{20.75} s \\
        	hMETIS-R && \numprint{625} & \numprint{17.19} s && \numprint{1684} & \numprint{66.64} s \\
        hMETIS-K && \numprint{524} & \numprint{14.81} s && \numprint{1587} & \numprint{43.44} s \\
        PaToH && \numprint{508} & \numprint{0.46} s && \numprint{1679} & \numprint{0.51} s \\
        	Zoltan && \numprint{943} & \numprint{0.66} s && \numprint{1962} & \numprint{1.31} s \\
        \midrule
        KaHIP-Strong && \numprint{517} & \numprint{84.33} s && \numprint{1638} & \numprint{191.36} s \\
        METIS && \numprint{595} & \numprint{0.63} s && \numprint{1829} & \numprint{1.92} s \\
        \midrule
        	Ja-Be-Ja-VC && \numprint{6336} & -- && \numprint{18441} & --  \\
        \bottomrule
    \end{tabularx}

    \caption{Mean vertex cuts and running times for \texttt{rhg10} -- \texttt{rhg18} (left) and Walshaw graphs and SPMV graphs with up to $1$M nodes (right).}
    \label{tab:li}
\end{table}

As can be seen in Figure~\ref{fig:plot_li} and Table~\ref{tab:li}, partitioning the hypergraph model with KaHyPar-MF overall results in the lowest vertex cuts. Moreover we see that all
hypergraph partitioners except Zoltan on average perform better than SPAC+METIS, with PaToH even being faster.
This is true not only for meshes and SPMV graphs, but also for rhgX graphs with power law degree distribution.
This effect could be explained by the choice of $k$---the number of blocks used for partitioning. While we use standard values for (node-based) graph
partitioning benchmarks~\cite{walshaw2000mpm}, Li~\etal~\cite{li2017spac} choose $k$ such that each block contains \emph{approximately} \numprint{10240} edges.
Thus some instances are partitioned into up to \numprint{1692} and \numprint{5952} blocks, which might be too large for current partitioning tools.

Looking at the solution quality of different SPAC+X approaches, we see that KaHIP performs better than METIS when using its \emph{strong}
configuration and even outperforms all hypergraph partitioners except KaHyPar-MF. However on Walshaw and SPMV graphs, it is also the slowest partitioning approach.

As can be seen in the additional performance plots in Appendix~\ref{app:perf_plots}, KaHyPar-MF dominates all other hypergraph partitioners in terms of solution quality.
Compared to hMETIS and PaToH, its solutions are \numprint{18}\%  and \numprint{19}\% better on average.
In a sequential edge partitioning setting with a reasonable number $k$ of blocks, we therefore conclude that hypergraph partitioning performs better than SPAC+X regarding both
solution quality (using KaHyPar-MF) and running time (using PaToH).
Finally we note that the distributed edge partitioner Ja-Be-Ja-VC can not compete with high quality graph or hypergraph partitioning systems. Since its solutions are more than an order of magnitude worse, we do not consider it in the following comparisons.

\subsection{Solution Quality of dSPAC+X and dHGP.}
\label{sec:distributed_split_graph_partitioning} 

We now investigate state-of-the-art methods for computing edge partitions in the distributed memory setting. Here, we use the large graphs from Table~\ref{tab:instances} including \texttt{rhg20} -- \texttt{rhg26}, as well
as the large two SPMV graphs \texttt{in-2004\_spmv} and \texttt{circuit5M\_smpv}.
Since Ja-Be-Ja-VC~\cite{fatemeh2014jabejavc} already produced low quality solutions on small graphs, we restrict the following comparison to distributed memory hypergraph partitioning with Zoltan
and distributed graph partitioning using our distributed split graph construction (dSPAC) in combination with both ParMETIS and ParHIP. All instances are
again partitioned into $k\in\{2,4,8,16,32,64,128\}$ blocks. This time, we run all algorithms on 32 cluster nodes (i.e., with 640 PEs in total).  

As can be seen in Table~\ref{tbl:performance_all_large} and Figure~\ref{fig:performance_all_large}, dSPAC-based graph partitioning outperforms
the hypergraph partitioning approach using Zoltan in \emph{both} solution quality and running time. While dSPAC+ParMETIS is the fastest configuration, dSPAC+ParHIP-Eco
provides the best solution quality. On irregular social networks and web graphs (type \emph{S}), the solutions of ParHIP-Fast and ParMETIS are on average \numprint{15}\% and \numprint{18}\% worse, respectively,
while the vertex cuts produced by Zoltan are worse by more than a factor of two. Furthermore dSPAC+X also performs better than HGP with Zoltan for meshes, road networks and SPMV graphs.
Since dSPAC+X outperforms HGP with Zoltan in a distributed setting regarding both solution quality \emph{and} running time, we thus conclude that is is currently the best
approach for computing edge partitions of large graphs, in particular if the graphs do not fit into the memory of a single machine.

\begin{table}
    \begin{tabularx}{\columnwidth}{lXrrXrr}
        \toprule
        Type && \multicolumn{2}{c}{\emph{S}} && \multicolumn{2}{c}{ \emph{M} \& \emph{R} \& \emph{SP}} \\
        Algorithm&& \multicolumn{1}{c}{VC} & \multicolumn{1}{c}{Time} && \multicolumn{1}{c}{VC} & \multicolumn{1}{c}{Time} \\
        \midrule
        ParHIP-Fast && \numprint{22321} & \numprint{17.92} s && \numprint{8380} & \numprint{9.54} s \\
        ParHIP-Eco &&  \numprint{18952} & \numprint{65.91} s && \numprint{7255} & \numprint{37.30} s \\
        ParMETIS &&  \numprint{23221} & \numprint{3.01} s && \numprint{9432} & \numprint{1.66} s \\
        \midrule
        Zoltan && \numprint{50780} & \numprint{51.82} s && \numprint{13516} & \numprint{51.80} s \\
        \bottomrule
    \end{tabularx}
    \caption{Mean vertex cuts and running times for large social graphs (left), as well as meshes, road networks, and SPMV graphs (right).}\label{tbl:performance_all_large}
\end{table}

\begin{figure}
	\includegraphics[width=\columnwidth]{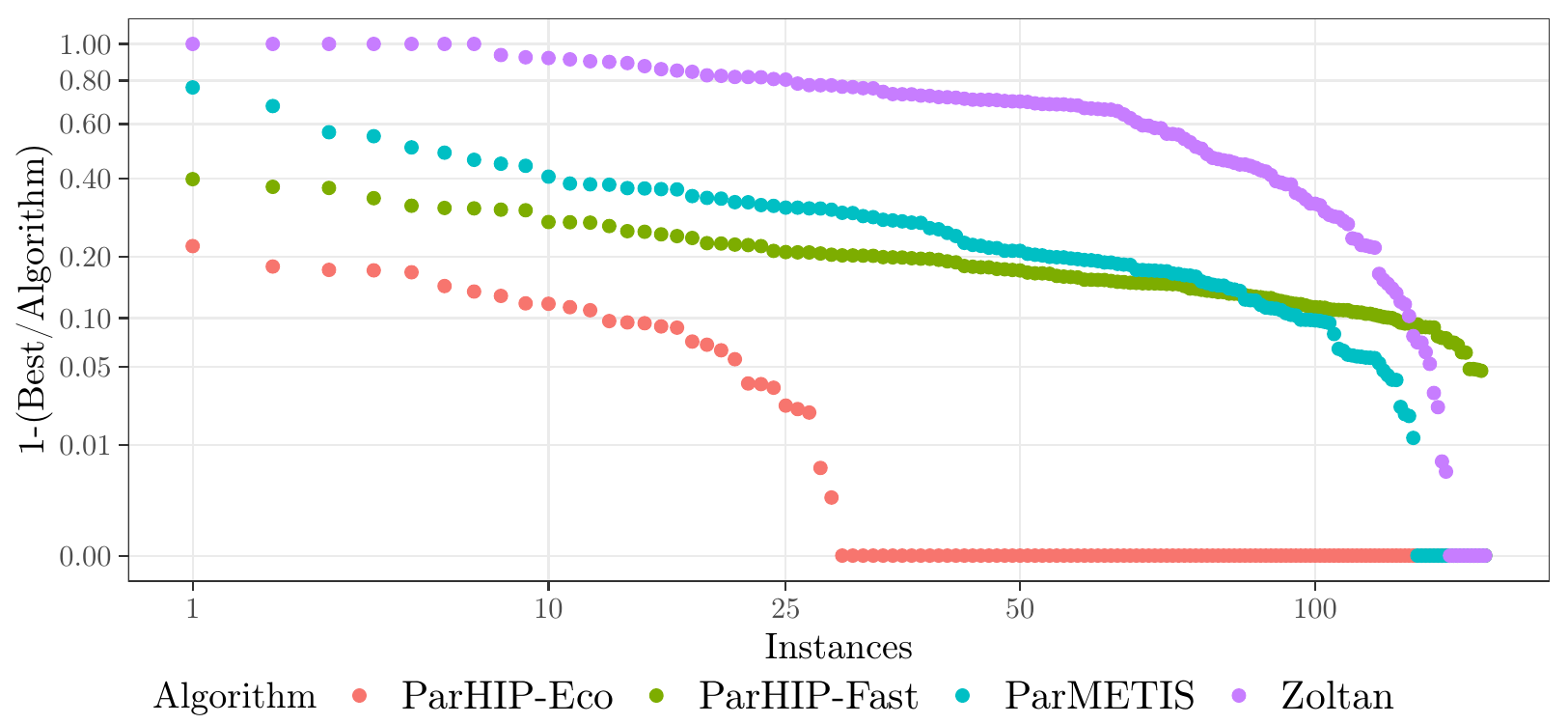}
	\caption{Solution quality for distributed approaches on large social graphs, meshes, road networks, and SPMV graphs. Instances\protect\footnotemark~where Zoltan required too much memory were set to $1$.}
	\label{fig:performance_all_large}	 
\end{figure}

\subsection{Scalability and Solution Quality of dSPAC+X.}
Finally, we look at the scaling behavior of distributed SPAC graph construction and partitioning using ParMetis and ParHIP-Fast. To simplify the
evaluation, we restrict the experiments in this section to partitioning the eight largest graphs (including \texttt{rgg25} - \texttt{rgg28}) into $k=2$ blocks on an increasing number of PEs.
We start with a single PE on a single node and then go up to all \numprint{20} PEs of a single node.
From there on we double the number of nodes in each step, until we arrive at \numprint{128} nodes with a total of \numprint{2560} PEs.
The results  are shown in Figure~\ref{fig:scaling}. Results for the remaining large graphs can be found in Figure~\ref{app:fig:scaling} in Appendix~\ref{app:dspac+x}.
The running times of dSPAC+X are dominated by the running times of the distributed graph partitioners.
While dSPAC+ParMetis is faster than dSPAC+ParHIP-Fast, the latter scales slightly better than the former.
Regarding solution quality, Table~\ref{app:cuts_large} in Appendix~\ref{app:dspac+x} shows that for large numbers of PEs, ParHIP-Fast computes better solutions
than ParMETIS.
By combining our distributed split graph construction algorithm with high quality distributed graph partitioning algorithm, we are now 
able to compute edge partitions of huge graphs that were previously not solvable on one a single PE, or even a small number of PEs.

\footnotetext{\texttt{nlpkkt240} with $k\in\{2,4,8,16,32,64,128\}$.}

\begin{figure}[t]
      \centering
              \begin{subfigure}[c]{\columnwidth}
        \includegraphics[width=\columnwidth]{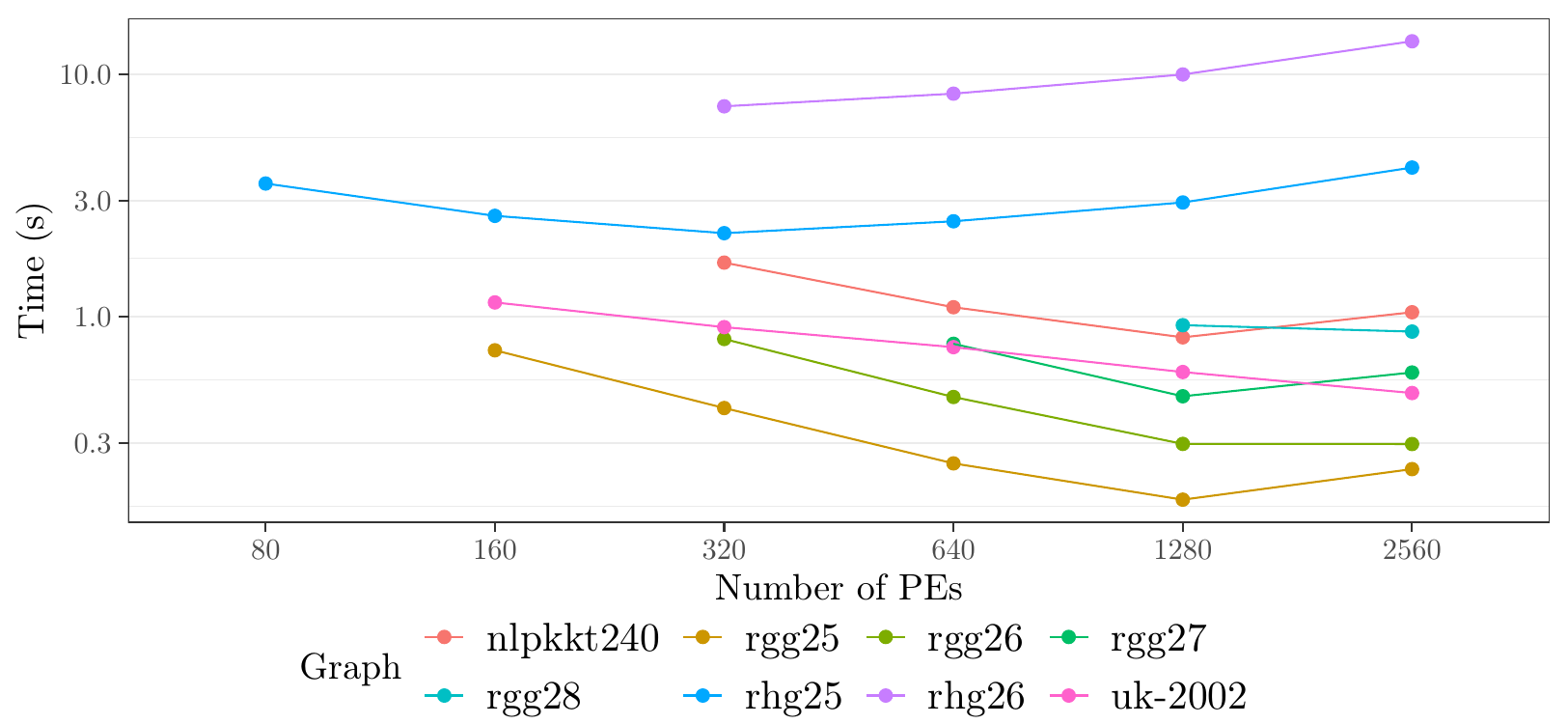}
        \caption{Running time for distributed split graph construction.}
        \label{fig:scale_split_large}
            \end{subfigure}
      \begin{subfigure}[c]{\columnwidth}
        \includegraphics[width=\columnwidth]{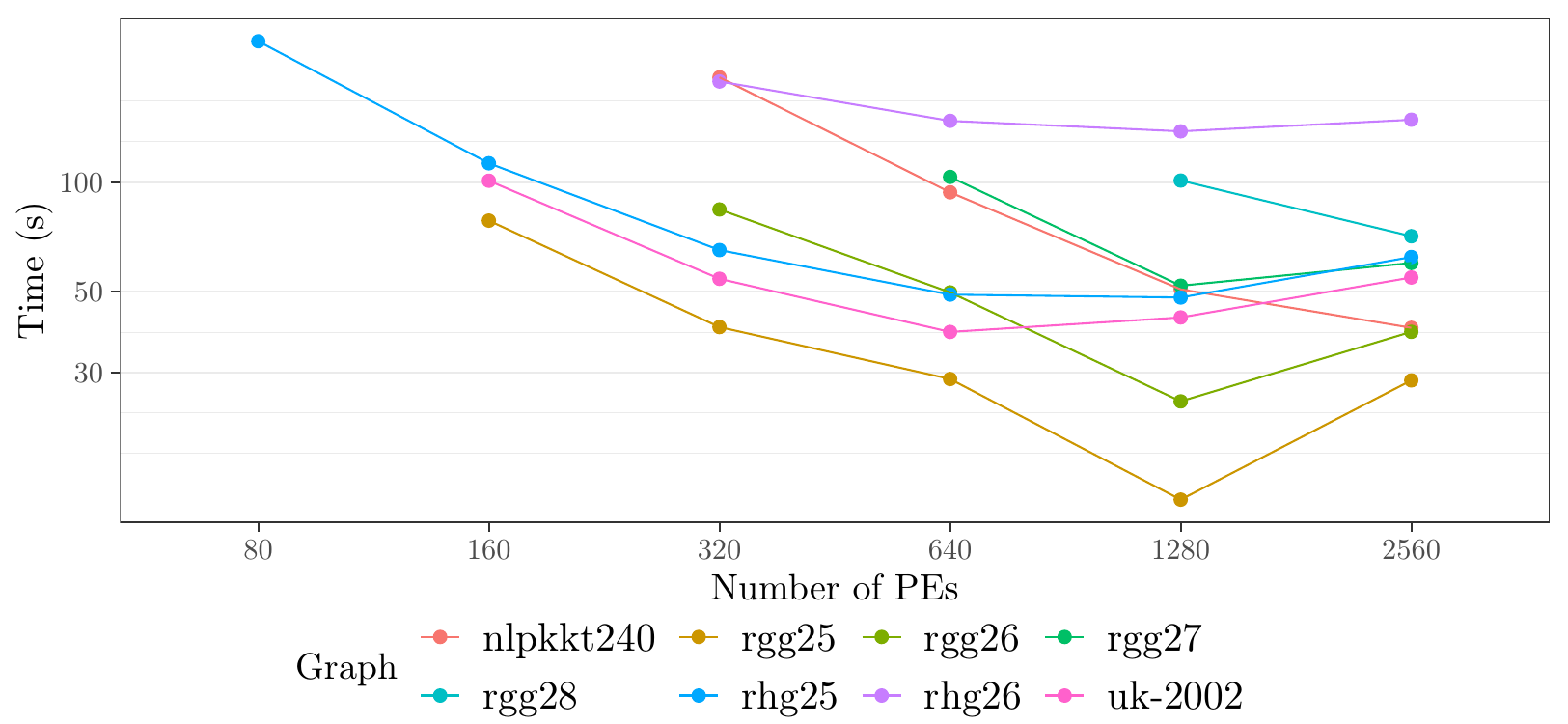}
        \caption{Running time for dSPAC+ParHIP-Fast.}
        \label{fig:scale_parhip_large}
    \end{subfigure}
          \begin{subfigure}[c]{\columnwidth}
        \includegraphics[width=\columnwidth]{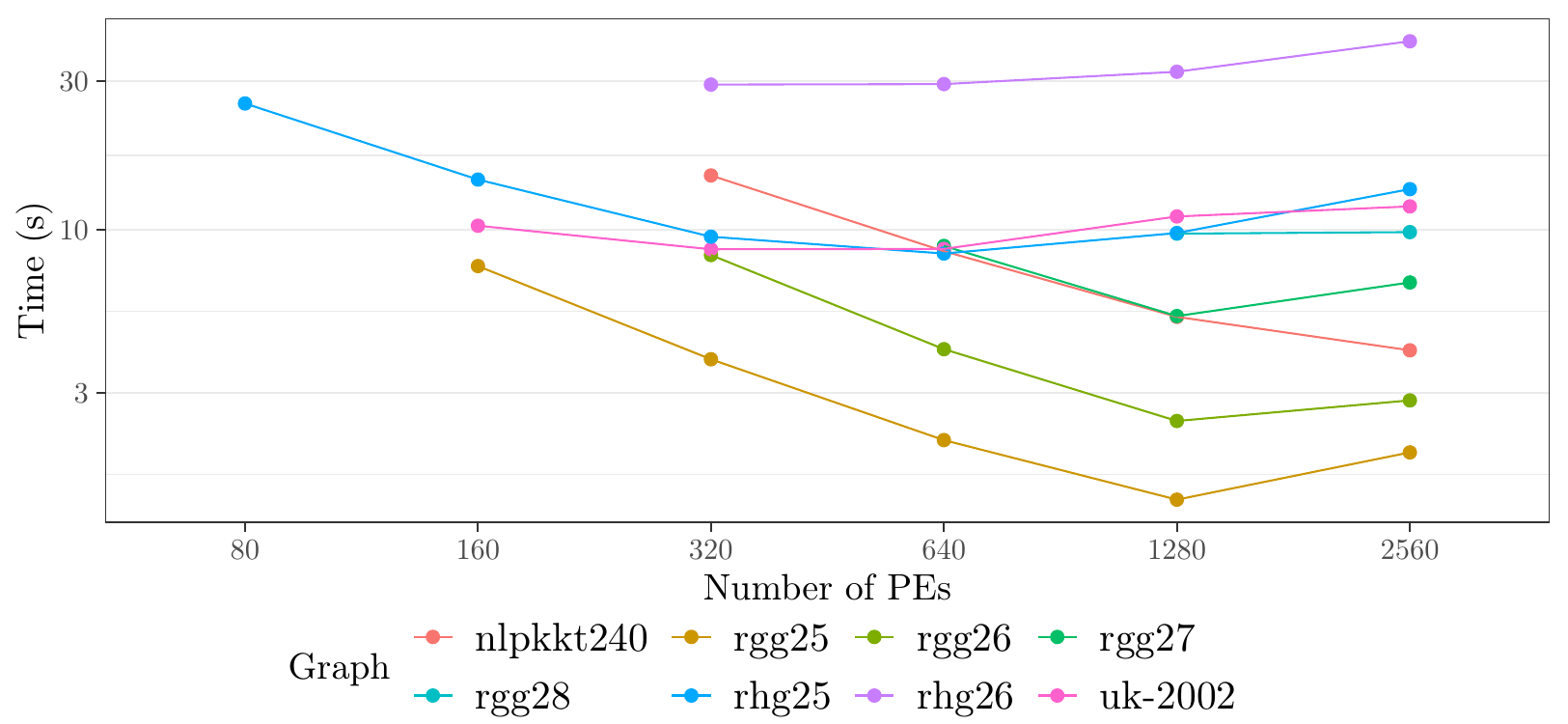}
        \caption{Running time for dSPAC+ParMETIS.}
        \label{fig:scale_parhip_large}
      \end{subfigure}
      \caption{Comparing the running times of distributed split graph construction (a), dSPAC+ParHIP-Fast (b), and dSPAC+ParMETIS (c) for the eight largest graphs of our benchmark set on
      an increasing number of PEs.}
      \label{fig:scaling}
\end{figure} 
\section{Conclusion and Future Work}
\label{sec:conclusion}
We presented an efficient distributed memory parallel edge partitioning algorithm that computes solutions of very high quality.
By efficiently parallelizing the split graph construction, our dSPAC+X algorithm scales to graphs with billions of edges and runs efficiently on up to 2560 PEs.
Our extensive experiments furthermore show that in a \emph{sequential} setting hypergraph partitioners still outperform node-based graph partitioning based on the SPAC approach
regarding both solution quality and running time. Hence, we believe that fast high-quality hypergraph partitioners yet have to be developed in order to narrow the gap between dSPAC+X and dHGP.

In the future, we would like to run a working implementation of Par$k$way in order to get a complete overview regarding the state-of-the-art in distributed HGP. 
Furthermore it would be interesting to combine ParHIP with the the shared memory parallel MT-KaHIP~\cite{2017arXiv171008231A} partitioner in order to get
a partitioner that uses shared memory parallelism within a cluster node, while cluster nodes themselves still work in a distributed
memory fashion.
Lastly, we plan to release our algorithm.

\subsection*{Acknowledgments.}
This work was performed on the computational resource ForHLR II funded by the Ministry of Science, Research and the Arts Baden-W\"urttemberg and DFG ("Deutsche Forschungsgemeinschaft").

\renewcommand\bibsection{\section*{\refname}}
\bibliographystyle{abbrvnat}
\bibliography{phdthesiscs}
\clearpage
\begin{appendix}

\onecolumn
\section{Full Experimental Data}
\label{appendix:data}
\begin{table}[h] 
	\centering
	\begin{subfigure}[t]{0.49\columnwidth}
    \begin{tabularx}{\columnwidth}{lXrXr}
        \toprule
        Algorithm && VC (avg) && Runtime \\
        \midrule
        KaHyPar-MF && \numprint{1283} && \numprint{17.30} s \\
        hMETIS-K && \numprint{1518} && \numprint{35.06} s \\
        hMETIS-R && \numprint{1617} && \numprint{51.85} s \\
        PaToH && \numprint{1604} && \numprint{0.42} s \\
        Zoltan (1) && \numprint{1875} && \numprint{1.06} s \\
        Zoltan (20) && \numprint{1967} && \numprint{0.25} s \\
        \midrule
        KaHIP-Fast && \numprint{1869} && \numprint{1.66} s \\
        KaHIP-Eco && \numprint{1663} && \numprint{7.76} s \\
        KaHIP-Strong && \numprint{1572} && \numprint{160.02} s \\
        METIS && \numprint{1750} && \numprint{1.55} s \\
        \midrule
        ParHIP-Fast (20) && \numprint{1947} && \numprint{0.66} s \\
        ParHIP-Eco (20) && \numprint{1737} && \numprint{108.26} s \\
        ParMETIS (20) && \numprint{1970} && \numprint{0.08} s \\
        \midrule
        Ja-Be-Ja-VC && \numprint{15077} && -- \\
        \bottomrule
    \end{tabularx}
    \caption{Walshaw graphs.}
    \label{tab:app_full_walshaw}
    \vspace{0.5em}
    \end{subfigure}
	\begin{subfigure}[t]{0.49\columnwidth}
    \begin{tabularx}{\columnwidth}{lXrXr}
        \toprule
        Algorithm && VC (avg) && Runtime \\
        \midrule
        KaHyPar-MF && \numprint{1283} && \numprint{162.84} s \\
        hMETIS-K && \numprint{2620} && \numprint{493.26} s \\
        hMETIS-R && \numprint{2664} && \numprint{1144.71} s \\
        PaToH && \numprint{2813} && \numprint{4.69} s \\
        Zoltan (1) && \numprint{3289} && \numprint{14.39} s \\
        Zoltan (20) && \numprint{3874} && \numprint{3.40} s \\
        \midrule
        KaHIP-Fast && \numprint{3120} && \numprint{11.21} s \\
        KaHIP-Eco && \numprint{2754} && \numprint{52.53} s \\
        KaHIP-Strong && \numprint{2598} && \numprint{1428.63} s \\
        METIS && \numprint{3014} && \numprint{20.60} s \\
        \midrule
        ParHIP-Fast (20) && \numprint{4039} && \numprint{12.77} s \\
        ParHIP-Eco (20) && \numprint{3289} && \numprint{146.06} s \\
        ParMETIS (20) && \numprint{3403} && \numprint{0.63} s \\
        \midrule
        Ja-Be-Ja-VC && \numprint{180762} && -- \\
        \bottomrule
    \end{tabularx}
    \caption{SPMV graphs with up to $1$M edges.}
    \label{tab:app_full_spmv}
	\end{subfigure}
	\begin{subfigure}[t]{0.49\columnwidth}
    \begin{tabularx}{\columnwidth}{lXrXr}
        \toprule
        Algorithm && VC (avg) && Runtime \\
        \midrule
        KaHyPar-MF && \numprint{433} && \numprint{510.75} s \\
        hMETIS-K && \numprint{524} && \numprint{14.81} s \\
        hMETIS-R && \numprint{625} && \numprint{17.19} s \\
        PaToH && \numprint{508} && \numprint{0.46} s \\
        Zoltan (1) && \numprint{1301} && \numprint{0.36} s \\
        	Zoltan (20) && \numprint{1481} && \numprint{0.19} s \\
        \midrule
        KaHIP-Fast && \numprint{658} && \numprint{1.04} s \\
        KaHIP-Eco && \numprint{580} && \numprint{4.50} s \\
        KaHIP-Strong && \numprint{517} && \numprint{84.33} s \\
        METIS && \numprint{595} && \numprint{0.63} s \\
        \midrule
        ParHIP-Fast (20) && \numprint{613} && \numprint{0.52} s \\
        ParHIP-Eco (20) && \numprint{538} && \numprint{106.20} s \\
        ParMETIS (20) && \numprint{676} && \numprint{0.06} s \\
        \midrule
        Ja-Be-Ja-VC && \numprint{6336} && -- \\
        \bottomrule
    \end{tabularx}
    \caption{Small random hyperbolic graphs, \texttt{rhg10} -- \texttt{rhg18}.}
    \label{tab:app_full_rhg}
    \end{subfigure}
	\begin{subfigure}[t]{0.49\columnwidth}
    \begin{tabularx}{\columnwidth}{lXrXr}
        \toprule
        Algorithm && VC (avg) && Runtime \\
        \midrule
        KaHyPar-MF && \numprint{1091} && \numprint{43.58} s \\
        hMETIS-K && \numprint{1321} && \numprint{37.99} s \\
        hMETIS-R && \numprint{1443} && \numprint{55.00} s \\
        PaToH && \numprint{1342} && \numprint{0.54} s \\
        Zoltan (1) && \numprint{1846} && \numprint{1.06} s \\
        Zoltan (20) && \numprint{1979} && \numprint{0.30} s \\
        \midrule
        KaHIP-Fast && \numprint{1593} && \numprint{1.81} s \\
        KaHIP-Eco && \numprint{1412} && \numprint{8.43} s \\
        KaHIP-Strong && \numprint{1317} && \numprint{173.79} s \\
        METIS && \numprint{1480} && \numprint{1.64} s \\
        \midrule
        ParHIP-Fast (20) && \numprint{1651} && \numprint{0.80} s \\
        ParHIP-Eco (20) && \numprint{1459} && \numprint{110.39} s \\
        ParMETIS (20) && \numprint{1679} && \numprint{0.08} s \\
        \midrule
        Ja-Be-Ja-VC && \numprint{15774} && -- \\
        \bottomrule
    \end{tabularx}
    \caption{All small graphs.}
    \label{tab:app_full_small}
    \end{subfigure}
    \caption{Full comparison of Walshaw, SPMV graphs with up to $1$M nodes and \texttt{rhg10} -- \texttt{rhg18}. The results for Zoltan are reported for one PE, Zoltan (1), and 20 PEs, Zoltan (20). For ParHIP and ParMETIS, only results on 20 PEs are shown. All algorithms were executed on a single node.}
\end{table}

\clearpage

\section{Additional Performance Plots}\label{app:perf_plots}
The following performance plots are based on the Walshaw graphs,
          SPMV graphs with up to 1M nodes and the random hyperbolic graphs \texttt{rhg10} -- \texttt{rhg18}.
          
\begin{figure}[h]
	\begin{subfigure}[t]{0.5\columnwidth}
        \includegraphics[width=\columnwidth]{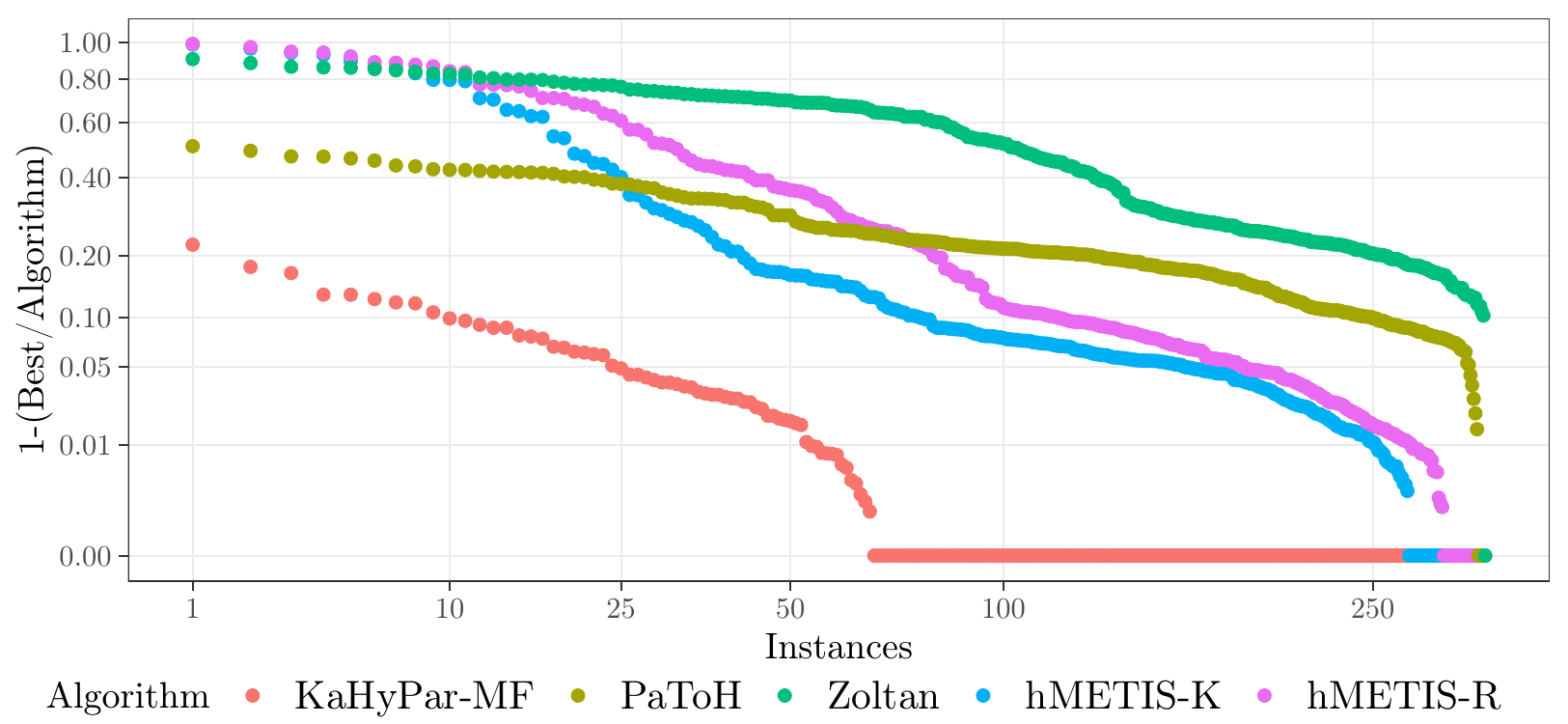}
        \caption{Performance comparison of all hypergraph partitioners.}
        \label{fig:performance_walshaw_spmv_rhg_hgr}
	\end{subfigure}
	\begin{subfigure}[t]{0.5\columnwidth}
		\includegraphics[width=\columnwidth]{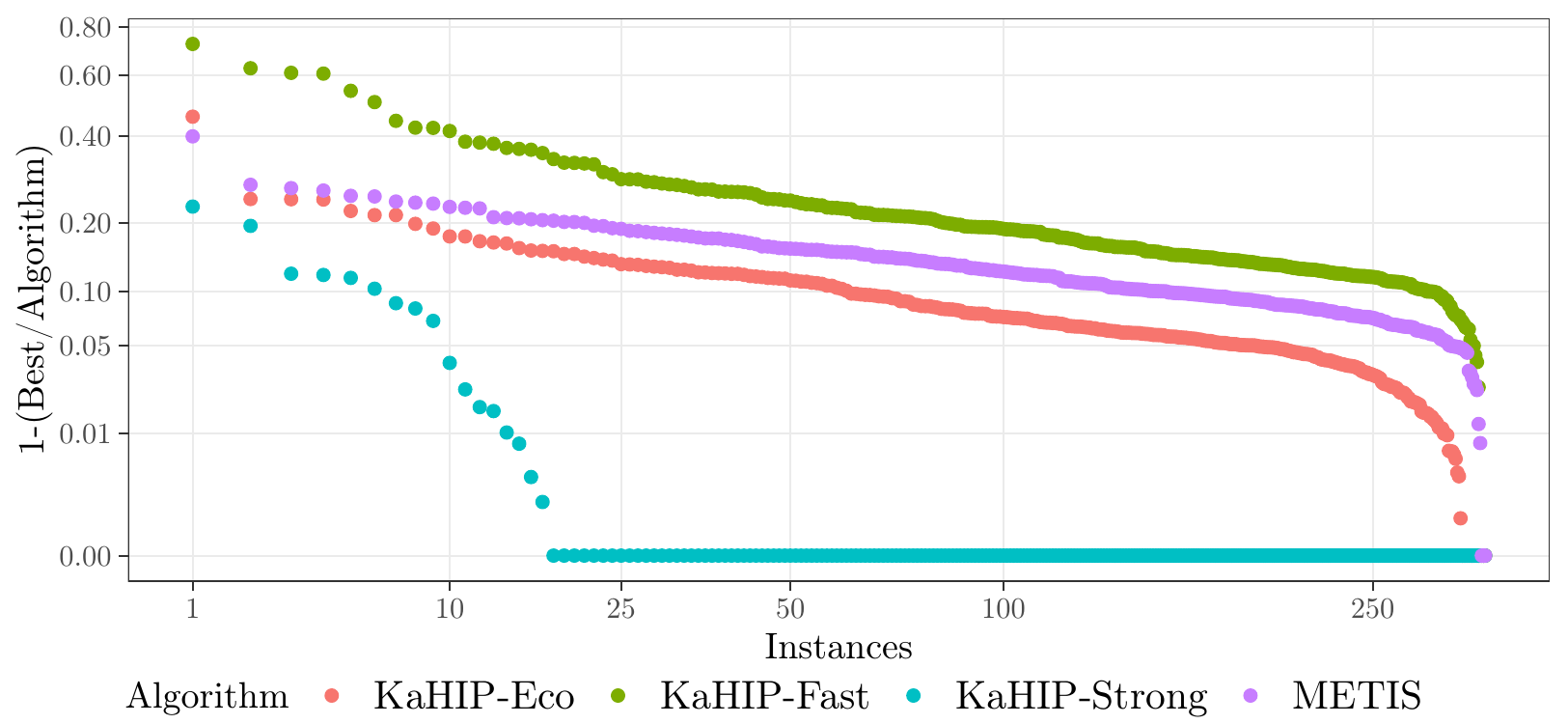}
		\caption{Performance comparison of SPAC+X with different sequential graph partitioners.}
        \label{fig:performance_walshaw_spmv_rhg_seqep}
	\end{subfigure}
	\begin{subfigure}[t]{0.5\columnwidth}
		\includegraphics[width=\columnwidth]{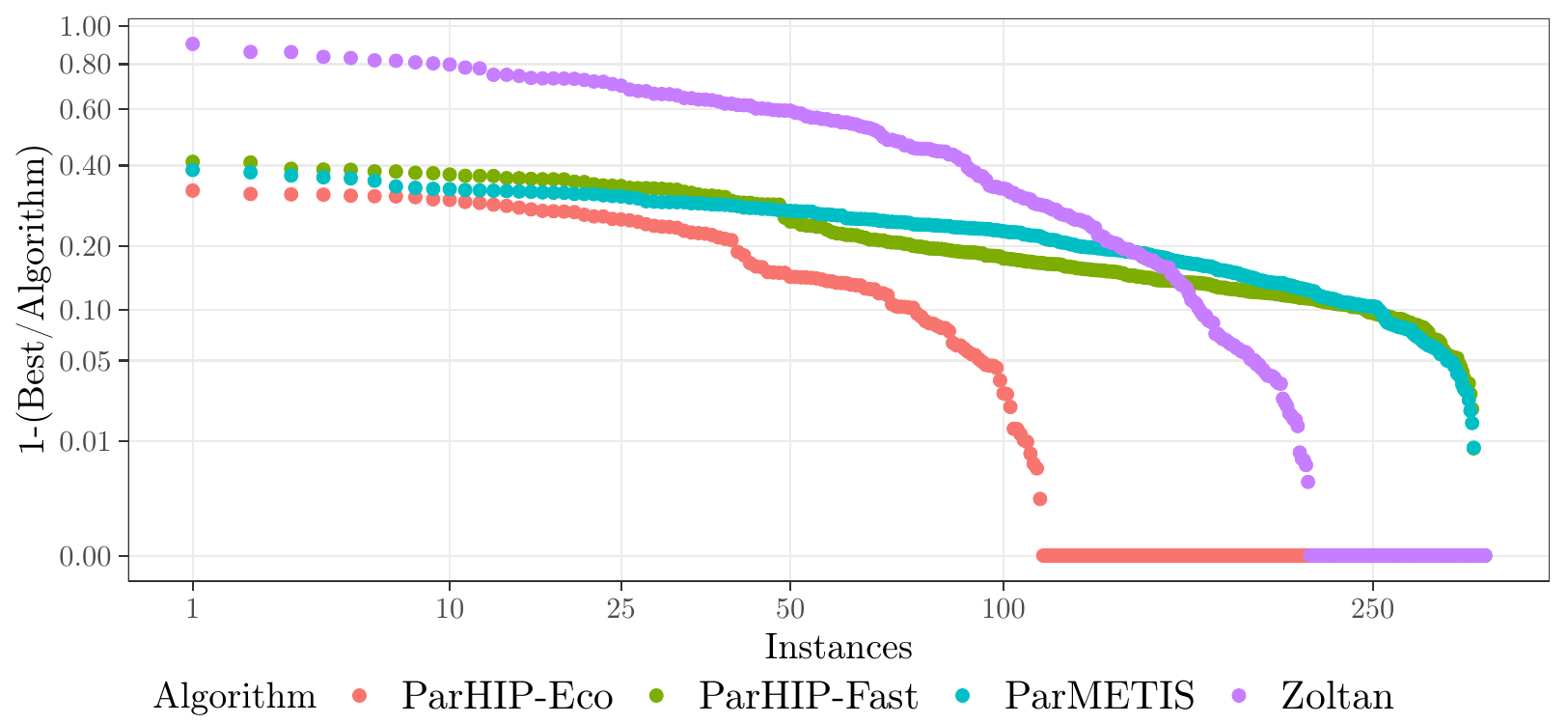}
     	\caption{Performance comparison of dSPAC+X with different distributed graph partitioners and Zoltan.}
        \label{fig:performance_walshaw_spmv_rhg_parep}
    \end{subfigure}
    \begin{subfigure}[t]{0.5\columnwidth}
		\includegraphics[width=\columnwidth]{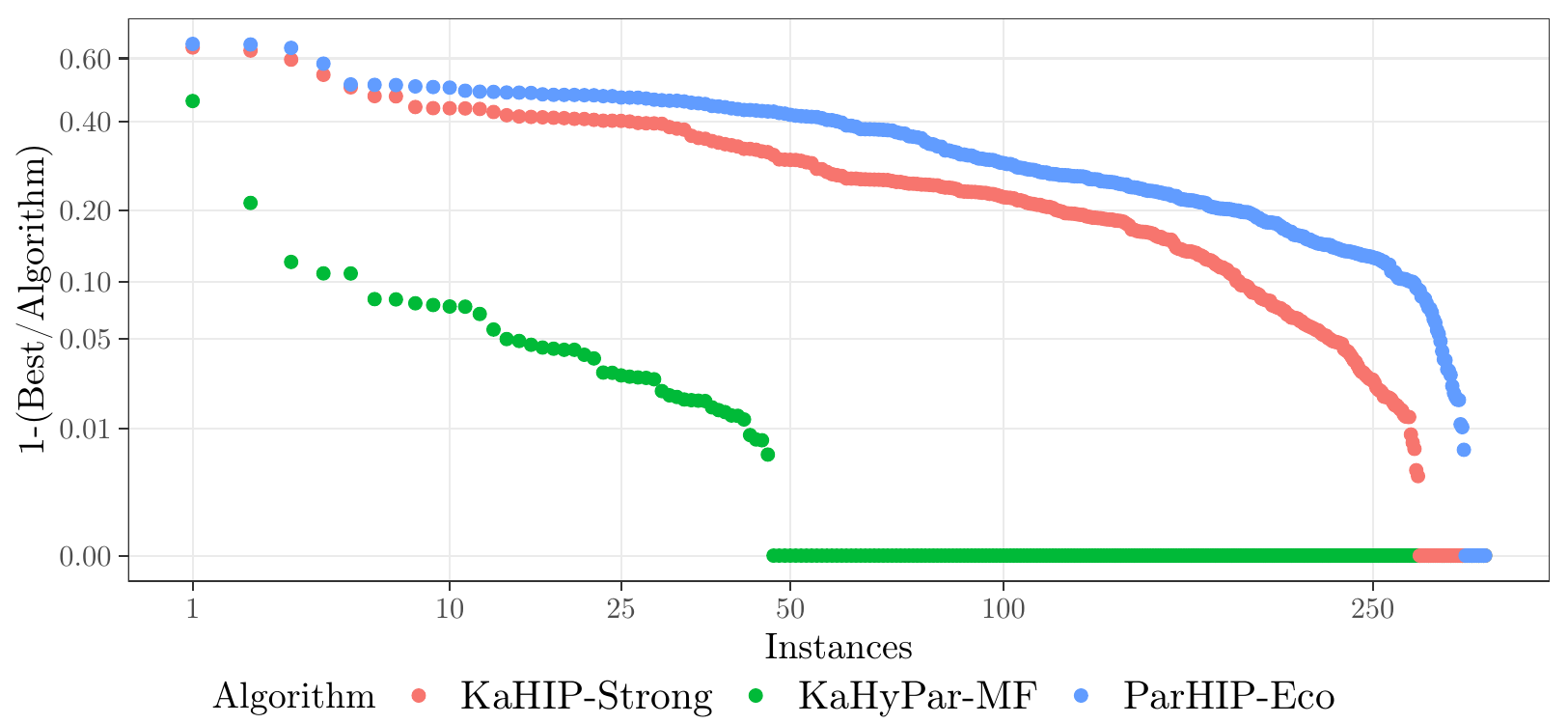}
		\caption{Comparing the best sequential SPAC+X (KaHIP-Strong) and HGP (KaHyPar-MF) approaches with the best dSPAC+X implementation (dSPAC+ParHIP-Eco).}
		\label{fig:all}	
	\end{subfigure}
	\caption{Additional performance plots comparing hypergraph partitioning algorithms (a), sequential SPAC+X approaches (b), dSPAC+X and dHGP approaches (c), as well as the best algorithm of each category (d).}
\end{figure}

\clearpage

\section{Additional Experimental Results for dSPAC+X}\label{app:dspac+x}
\begin{table}[h]
 	\centering
 	\begin{tabular}[10]{|lrrrrrrrrr|}
 \hline
 Graph                 & \multicolumn{9}{c|}{Vertex cut on different numbers of PEs}                                                                                                                                            \\
                       & \numprint{1}     & \numprint{20}    & \numprint{40}    & \numprint{80}    & \numprint{160}    & \numprint{320}    & \numprint{640}    & \numprint{1280}   & \numprint{2560}   \\
 \hline
 \hline
 \multicolumn{10}{|c|}{ParHIP-Fast}                                                                                                                                                                    \\
 \hline
\texttt{in-2004\_spmv} & \numprint{232}   & \numprint{366}   & \numprint{307}   & \numprint{235}   & \numprint{202}    & \numprint{195}    & \numprint{194}    & \numprint{189}    & \numprint{185}    \\
\texttt{circuit5M\_spmv} & --& --& \numprint{1526} & \numprint{1529} & \numprint{1544} & \numprint{1363} & \numprint{1316} & \numprint{1426} & \numprint{1475} \\
\texttt{amazon}        & \numprint{22740} & \numprint{23501} & \numprint{25432} & \numprint{23537} & \numprint{22987}  & \numprint{23577}  & \numprint{23364}  & \numprint{24493}  & \numprint{22788}  \\
\texttt{eu-2005}       & \numprint{1658}  & \numprint{1600}  & \numprint{1478}  & \numprint{1662}  & \numprint{1533}   & \numprint{1562}   & \numprint{1560}   & \numprint{1489}   & \numprint{1435}   \\
\texttt{youtube}       & \numprint{60838} & \numprint{63129} & \numprint{72033} & \numprint{63284} & \numprint{63018}  & \numprint{65684}  & \numprint{63535}  & \numprint{67324}  & \numprint{66867}  \\
\texttt{in-2004}       & \numprint{412}   & \numprint{333}   & \numprint{394}   & \numprint{472}   & \numprint{417}    & \numprint{359}    & \numprint{264}    & \numprint{248}    & \numprint{263}    \\
\texttt{packing}       & \numprint{3137}  & \numprint{3324}  & \numprint{3179}  & \numprint{3282}  & \numprint{3335}   & \numprint{3193}   & \numprint{3423}   & \numprint{3494}   & \numprint{4272}   \\
\texttt{channel}       & --                & \numprint{14575} & \numprint{16551} & \numprint{26542} & \numprint{21975}  & \numprint{19212}  & \numprint{14996}  & \numprint{14779}  & \numprint{15233}  \\
\texttt{road\_central} & \numprint{118}   & \numprint{137}   & \numprint{131}   & \numprint{125}   & \numprint{128}    & \numprint{126}    & \numprint{125}    & \numprint{132}    & \numprint{122}    \\
\texttt{hugebubble-10} & \numprint{1699}  & \numprint{1934}  & \numprint{1836}  & \numprint{1798}  & \numprint{1783}   & \numprint{1773}   & \numprint{1770}   & \numprint{1760}   & \numprint{1752}   \\
\texttt{uk-2002}       & --                & --                & --                & --                & \numprint{130548} & \numprint{110529} & \numprint{126918} & \numprint{115485} & \numprint{135790} \\
\texttt{nlpkkt240}     & --                & --                & --                & --                & --                 & \numprint{187071} & \numprint{181478} & \numprint{178162} & \numprint{181843} \\
\texttt{europe\_osm}   & --                & \numprint{196}   & \numprint{194}   & \numprint{201}   & \numprint{195}    & \numprint{199}    & \numprint{195}    & \numprint{205}    & \numprint{192}    \\
\texttt{rhg20}         & \numprint{567}   & \numprint{890}   & \numprint{700}   & \numprint{756}   & \numprint{657}    & \numprint{637}    & \numprint{580}    & \numprint{581}    & \numprint{590}    \\
\texttt{rhg21}         & \numprint{850}   & \numprint{1361}  & \numprint{1077}  & \numprint{1037}  & \numprint{924}    & \numprint{878}    & \numprint{846}    & \numprint{821}    & \numprint{759}    \\
\texttt{rhg22}         & \numprint{1215}  & \numprint{1809}  & \numprint{1478}  & \numprint{1269}  & \numprint{1126}   & \numprint{1044}   & \numprint{1011}   & \numprint{1020}   & \numprint{965}    \\
\texttt{rhg23}         & --                & \numprint{1960}  & \numprint{2061}  & \numprint{1465}  & \numprint{1503}   & \numprint{1399}   & \numprint{1480}   & \numprint{1304}   & \numprint{1394}   \\
\texttt{rhg24}         & --                & --                & \numprint{2822}  & \numprint{2795}  & \numprint{2688}   & \numprint{2131}   & \numprint{2056}   & \numprint{1907}   & \numprint{1726}   \\
\texttt{rhg25}         & --                & --                & --                & \numprint{3186}  & \numprint{3014}   & \numprint{2535}   & \numprint{2811}   & \numprint{2380}   & \numprint{2324}   \\
\texttt{rhg26}         & --                & --                & --                & --                & --                 & \numprint{3741}   & \numprint{3416}   & \numprint{3281}   & \numprint{3036}   \\
\texttt{rgg27}         & --                & --                & --                & --                & --                 & --                 & \numprint{24213}  & \numprint{22547}  & \numprint{21810}  \\
\texttt{rgg25}         & --                & --                & --                & --                & \numprint{10710}  & \numprint{9611}   & \numprint{8804}   & \numprint{8058}   & \numprint{7777}   \\
\texttt{rgg26}         & --                & --                & --                & --                & --                 & \numprint{16041}  & \numprint{13133}  & \numprint{13140}  & \numprint{12581}  \\
\texttt{rgg28}         & --                & --                & --                & --                & --                 & --                 & --                 & \numprint{17897}  & \numprint{17846}  \\
 \hline 
 \hline
 \multicolumn{10}{|c|}{ParMETIS}                                                                                                                                                                       \\
 \hline
\texttt{in-2004\_spmv} & \numprint{400}   & \numprint{588}   & \numprint{509}   & \numprint{566}   & \numprint{558}    & \numprint{513}    & \numprint{516}    & \numprint{494}    & \numprint{513}    \\
\texttt{circuit5M\_spmv} & --& --& \numprint{1346} & \numprint{1314} & \numprint{1445} & \numprint{1374} & \numprint{1390} & \numprint{1453} & \numprint{1563} \\
\texttt{amazon}        & \numprint{19219} & \numprint{22114} & \numprint{21619} & \numprint{21010} & \numprint{20417}  & \numprint{20358}  & \numprint{20269}  & \numprint{20338}  & \numprint{20174}  \\
\texttt{eu-2005}       & \numprint{2453}  & \numprint{2696}  & \numprint{2339}  & \numprint{2405}  & \numprint{2455}   & \numprint{2855}   & \numprint{2618}   & \numprint{2554}   & \numprint{2756}   \\
\texttt{youtube}       & \numprint{56207} & \numprint{62012} & \numprint{61790} & \numprint{61342} & \numprint{61200}  & \numprint{61151}  & \numprint{60899}  & \numprint{61007}  & \numprint{61249}  \\
\texttt{in-2004}       & \numprint{586}   & \numprint{804}   & \numprint{787}   & \numprint{753}   & \numprint{774}    & \numprint{733}    & \numprint{740}    & \numprint{687}    & \numprint{707}    \\
\texttt{packing}       & \numprint{2962}  & \numprint{3567}  & \numprint{3472}  & \numprint{3602}  & \numprint{3469}   & \numprint{3323}   & \numprint{3327}   & \numprint{3268}   & \numprint{3169}   \\
\texttt{channel}       & \numprint{10995} & \numprint{12057} & \numprint{12185} & \numprint{12147} & \numprint{12209}  & \numprint{12418}  & \numprint{12247}  & \numprint{12145}  & \numprint{12054}  \\
\texttt{road\_central} & \numprint{240}   & \numprint{223}   & \numprint{222}   & \numprint{221}   & \numprint{210}    & \numprint{201}    & \numprint{224}    & \numprint{234}    & \numprint{229}    \\
\texttt{hugebubble-10} & \numprint{1710}  & \numprint{1773}  & \numprint{1831}  & \numprint{1779}  & \numprint{1775}   & \numprint{1746}   & \numprint{1729}   & \numprint{1718}   & \numprint{1713}   \\
\texttt{uk-2002}       & --                & --                & --                & --                & \numprint{72346}  & \numprint{71476}  & \numprint{71084}  & \numprint{71689}  & \numprint{72477}  \\
\texttt{nlpkkt240}     & --                & --                & --                & --                & --                 & \numprint{143048} & \numprint{141947} & \numprint{142466} & \numprint{141491} \\
\texttt{europe\_osm}   & --                & \numprint{257}   & \numprint{257}   & \numprint{240}   & \numprint{239}    & \numprint{241}    & \numprint{248}    & \numprint{246}    & \numprint{233}    \\
\texttt{rhg20}         & \numprint{475}   & \numprint{652}   & \numprint{649}   & \numprint{633}   & \numprint{647}    & \numprint{647}    & \numprint{619}    & \numprint{582}    & \numprint{583}    \\
\texttt{rhg21}         & \numprint{646}   & \numprint{839}   & \numprint{934}   & \numprint{889}   & \numprint{868}    & \numprint{856}    & \numprint{816}    & \numprint{777}    & \numprint{809}    \\
\texttt{rhg22}         & \numprint{761}   & \numprint{1175}  & \numprint{1116}  & \numprint{1217}  & \numprint{1094}   & \numprint{1086}   & \numprint{1062}   & \numprint{1106}   & \numprint{1035}   \\
\texttt{rhg23}         & \numprint{1297}  & \numprint{1598}  & \numprint{1625}  & \numprint{1554}  & \numprint{1666}   & \numprint{1705}   & \numprint{1598}   & \numprint{1545}   & \numprint{1459}   \\
\texttt{rhg24}         & --                & --                & \numprint{2249}  & \numprint{2294}  & \numprint{2277}   & \numprint{2252}   & \numprint{2112}   & \numprint{2144}   & \numprint{2001}   \\
\texttt{rhg25}         & --                & --                & --                & \numprint{2876}  & \numprint{2942}   & \numprint{2952}   & \numprint{2746}   & \numprint{2693}   & \numprint{2755}   \\
\texttt{rhg26}         & --                & --                & --                & --                & --                 & \numprint{3923}   & \numprint{3902}   & \numprint{3785}   & \numprint{3669}   \\
\texttt{rgg27}         & --                & --                & --                & --                & --                 & --                 & \numprint{20141}  & \numprint{19672}  & \numprint{20208}  \\
\texttt{rgg25}         & --                & --                & --                & --                & \numprint{8832}   & \numprint{8882}   & \numprint{8999}   & \numprint{8782}   & \numprint{8725}   \\
\texttt{rgg26}         & --                & --                & --                & --                & --                 & \numprint{13373}  & \numprint{13405}  & \numprint{13097}  & \numprint{12880}  \\
\texttt{rgg28}         & --                & --                & --                & --                & --                 & --                 & --                 & \numprint{29864}  & \numprint{29742}  \\
 	\hline
 	\end{tabular}
 	\caption{Solution quality of dSPAC+X for the largest graphs on an increasing number of PEs.}\label{app:cuts_large}
\end{table}

\begin{figure}[h!]
  \centering
     \begin{subfigure}[t]{0.49\columnwidth}
        \includegraphics[width=\columnwidth]{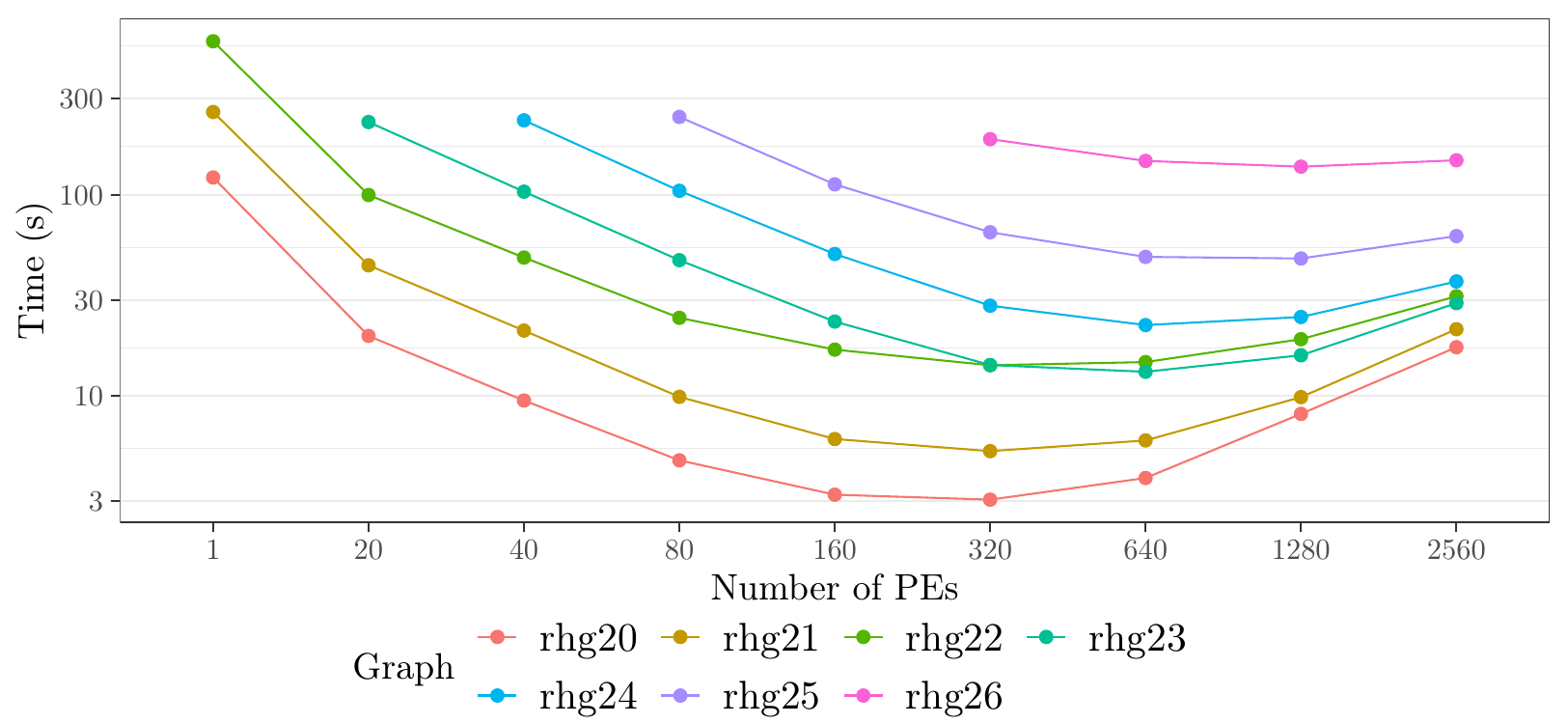}
        \caption{Running time for dSPAC+ParHIP-Fast.}
        \label{fig:scale_parhip_rhg}
    \end{subfigure}
    \begin{subfigure}[t]{0.49\columnwidth}
        \includegraphics[width=\columnwidth]{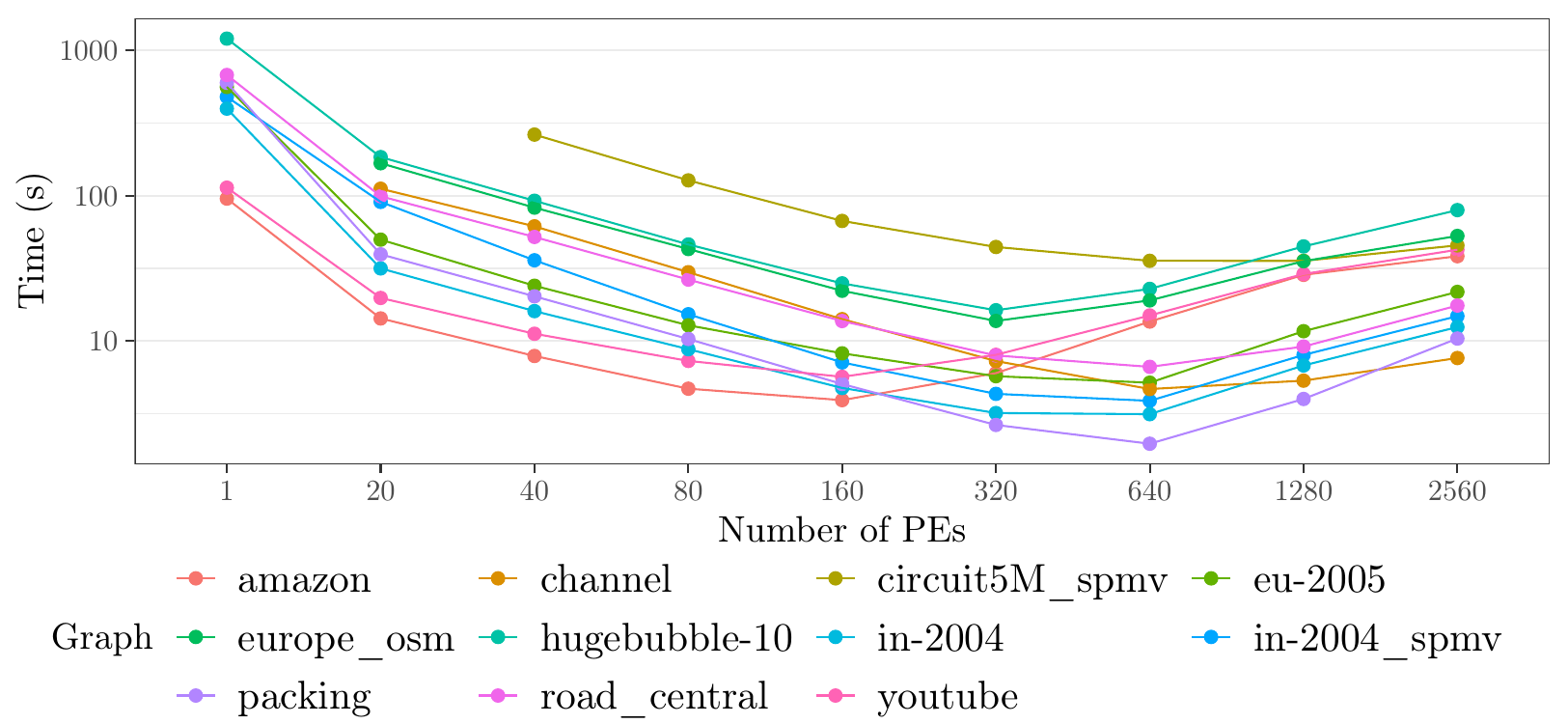}
        \caption{Running time for dSPAC+ParHIP-Fast.}
        \label{fig:scale_parhip_small}
    \end{subfigure}
    \begin{subfigure}[t]{0.49\columnwidth}
        \includegraphics[width=\columnwidth]{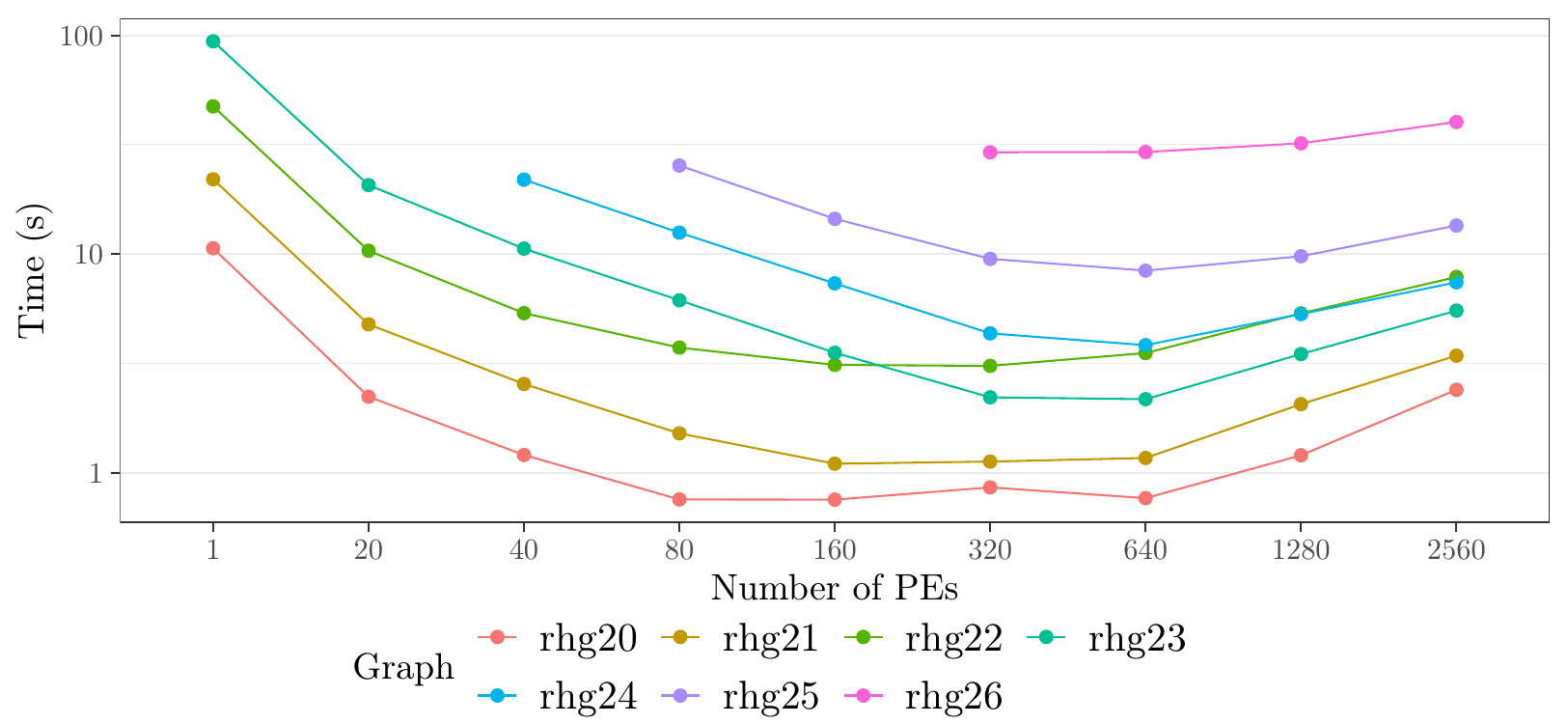}
        \caption{Running time for dSPAC+ParMETIS.}
        \label{fig:scale_parmetis_rhg}
    \end{subfigure}   
    \begin{subfigure}[t]{0.49\columnwidth}
        \includegraphics[width=\columnwidth]{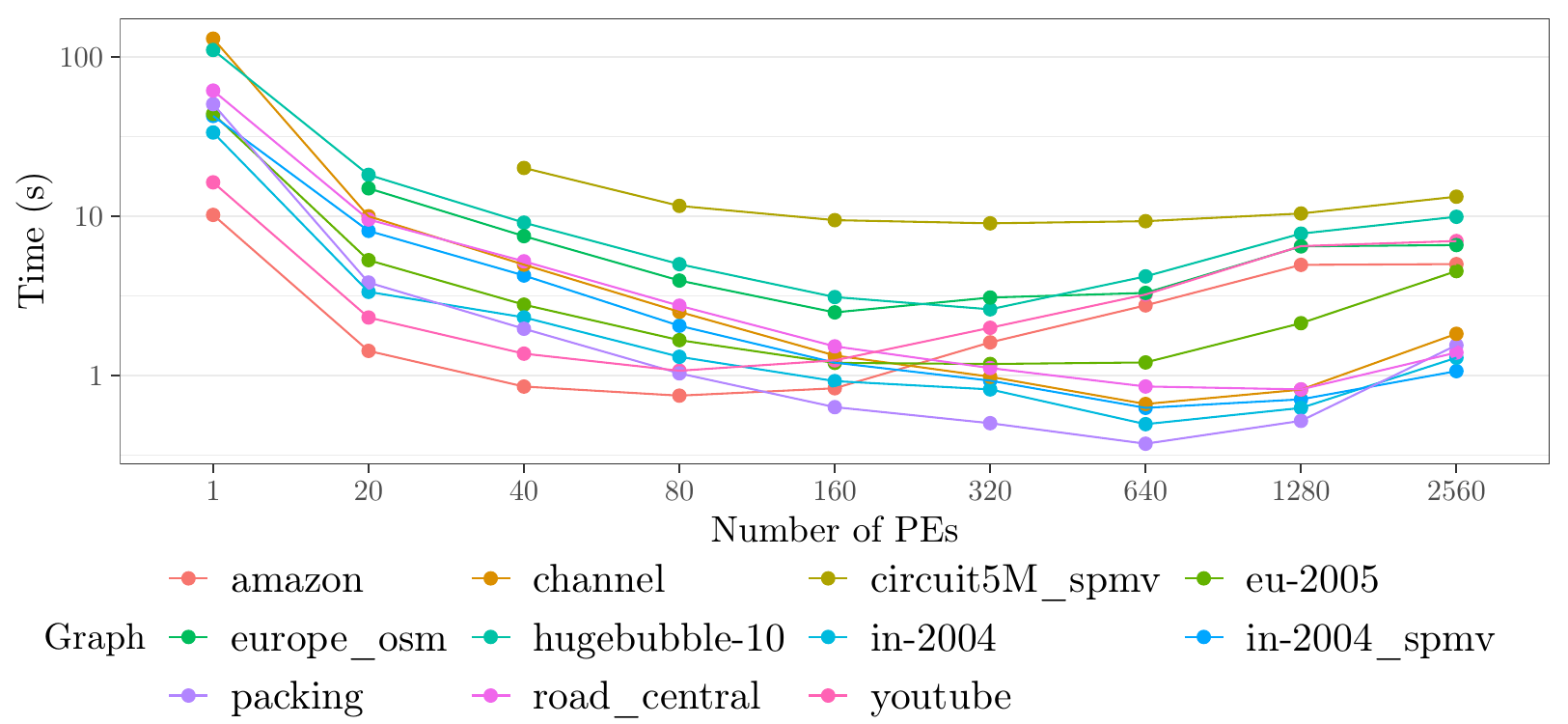}
        \caption{Running time for dSPAC+ParMETIS.}
        \label{fig:scale_parmetis_small}
    \end{subfigure}
    \begin{subfigure}[t]{0.49\columnwidth}
        \includegraphics[width=\columnwidth]{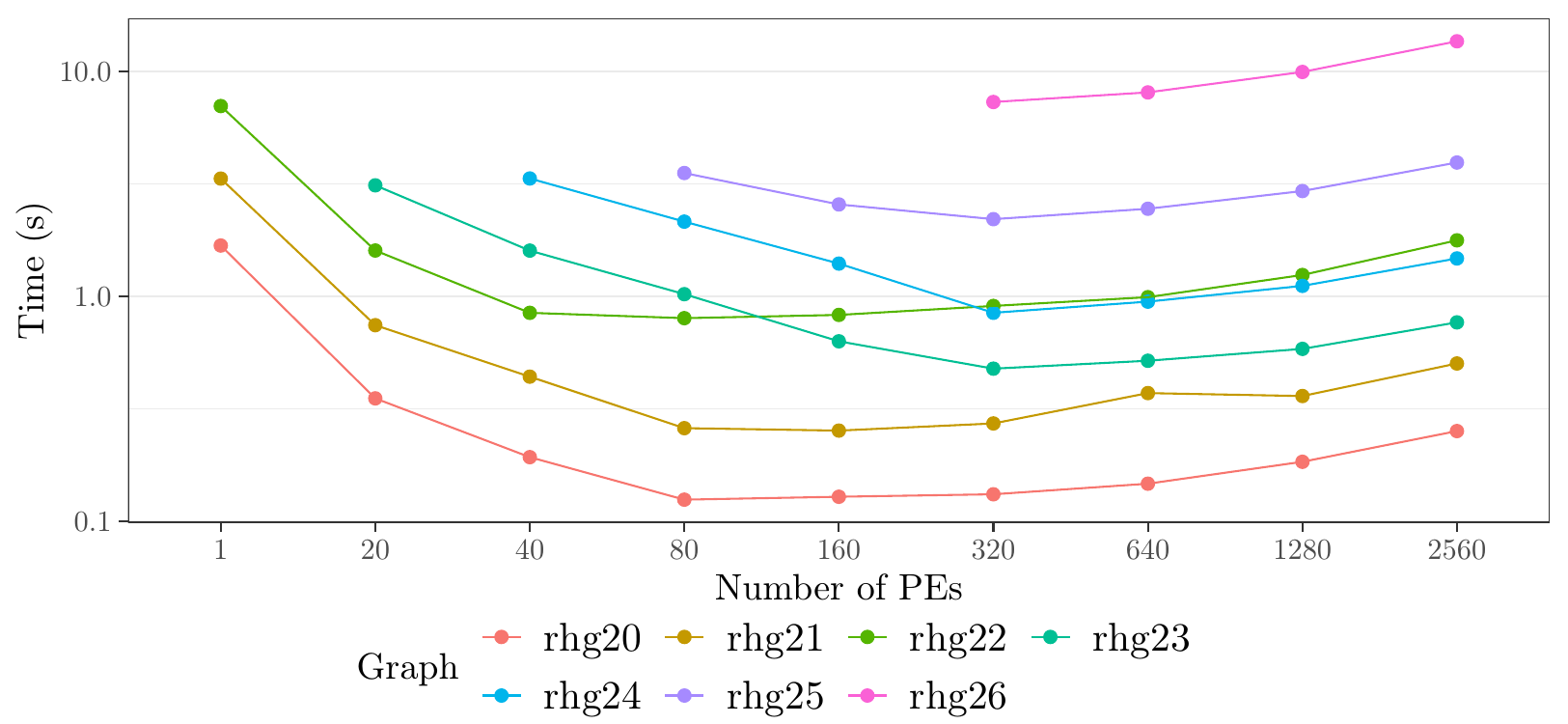}
        \caption{Running time for distributed split graph construction.}
        \label{fig:scale_split_rhg}
    \end{subfigure}
	\begin{subfigure}[t]{0.49\columnwidth}
        \includegraphics[width=\columnwidth]{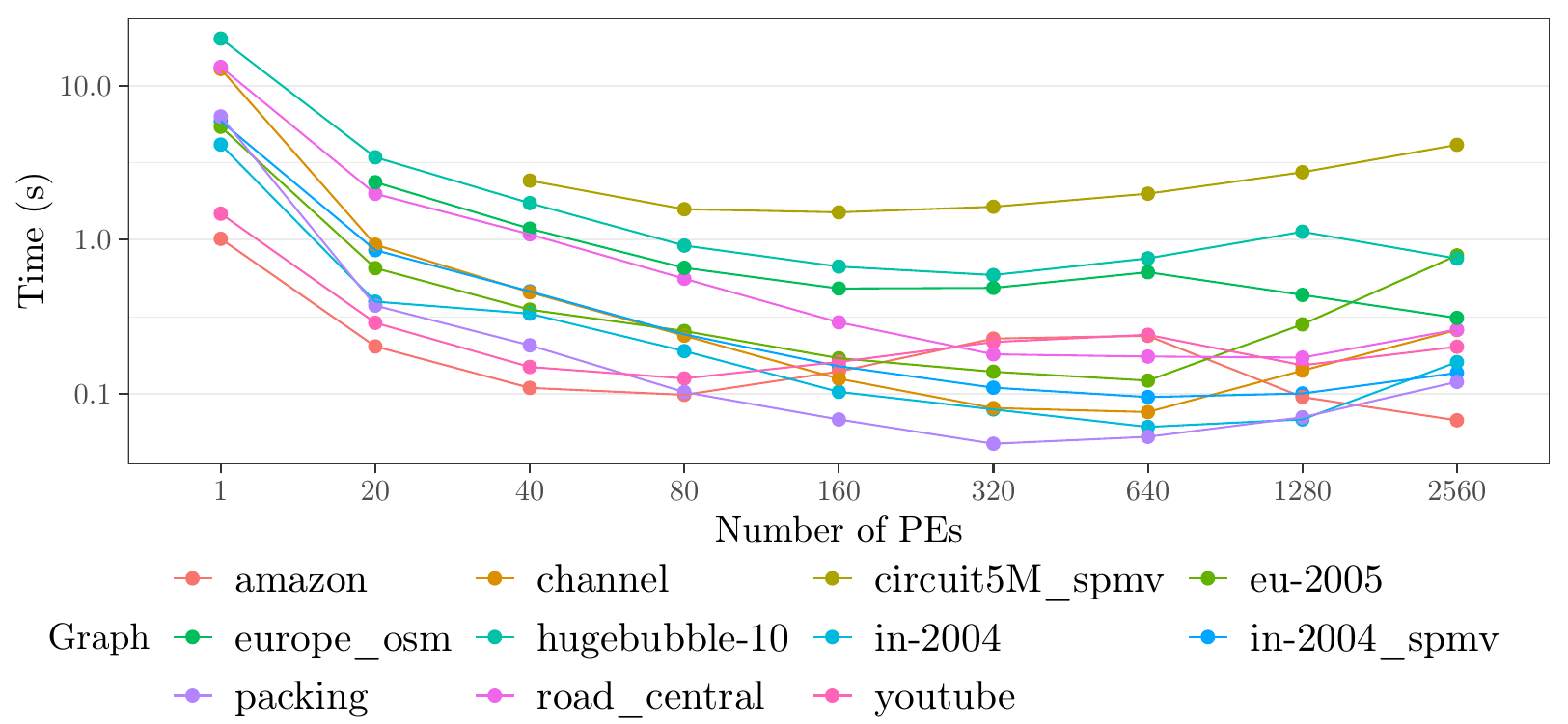}
        \caption{Running time for distributed split graph construction.}
        \label{fig:scale_split_small}
    \end{subfigure}
	\caption{Running times for dSPAC+ParHIP-Fast (a, b), dSPAC+ParMETIS (c, d) and the time it takes to compute the distributed split graph (e, f) on increasing numbers of PE.}
	\label{app:fig:scaling}
\end{figure}
\end{appendix}

\end{document}